\def\llncs{0}
\def\fullpage{1}
\def\anonymous{0}
\def\authnote{1}
\def\notxfont{0}
\def\submission{0}

\ifnum\submission=1
\def\llncs{1}
\fi

\ifnum\llncs=1 	\documentclass[envcountsect,a4paper,runningheads,10pt]{llncs}
\else
	\documentclass[letterpaper,hmargin=1.0in,vmargin=1.0in,11pt]{article}
			\ifnum\fullpage=1
		\usepackage{fullpage}
		\fi
\fi

\usepackage[%
  colorlinks=true,
  citecolor=blue,
  pagebackref=true
]{hyperref}

\usepackage{amsmath, amsfonts, amssymb, mathtools,amscd}

\usepackage{amsthm}

\usepackage{lmodern}
\usepackage[T1]{fontenc}
\usepackage[utf8]{inputenc}

\usepackage{arydshln} 
\usepackage{url}
\usepackage{ifthen}
\usepackage{bm}
\usepackage{multirow}
\usepackage[dvips]{graphicx}
\usepackage[usenames]{color}
\usepackage{xcolor,colortbl} 
\usepackage{threeparttable}
\usepackage{comment}
\usepackage{paralist,verbatim}
\usepackage{cases}
\usepackage{booktabs}
\usepackage{braket}
\usepackage{cancel} 
\usepackage{ascmac} 
\usepackage{framed}
\usepackage{authblk}
\usepackage{pifont}
\usepackage{qcircuit}
\usepackage{tikz}
\usetikzlibrary{arrows.meta}
\definecolor{darkblue}{rgb}{0,0,0.6}
\definecolor{darkgreen}{rgb}{0,0.5,0}
\definecolor{maroon}{rgb}{0.5,0.1,0.1}
\definecolor{dpurple}{rgb}{0.2,0,0.65}

\usepackage[capitalise,noabbrev]{cleveref}
\usepackage{aliascnt}
\usepackage[absolute]{textpos}
\usepackage[final]{microtype}
\usepackage[absolute]{textpos}
\usepackage{everypage}
\DeclareMathAlphabet{\mathpzc}{OT1}{pzc}{m}{it}

\usepackage{algorithmic}
\usepackage{algorithm}
\usepackage{here}

\newtheoremstyle{thicktheorem}%
{\topsep}
{\topsep}
{\itshape}{}%
{\bfseries}%
{.}
{ }%
{\thmname{#1}\thmnumber{ #2}%
		\thmnote{ (#3)}%
}

\newtheoremstyle{remark}
{\topsep}
{\topsep}
	{}
	{}
	{}
	{.}
	{ }
	{\textit{\thmname{#1}}\thmnumber{ #2}
			\thmnote{ (#3)}%
	}

\ifnum\llncs=0
	\theoremstyle{thicktheorem}
	\newtheorem{theorem}{Theorem}[section]
	\newtheorem{lemma}[theorem]{Lemma}
	\newtheorem{corollary}[theorem]{Corollary}

	\newaliascnt{definition}{theorem}
    \newtheorem{definition}[definition]{Definition}
    \aliascntresetthe{definition}

	\theoremstyle{remark}
	
	\newtheorem{remark}[theorem]{Remark}

\else
\fi

\Crefname{MyClaim}{Claim}{Claims}

	\crefname{theorem}{Theorem}{Theorems}
	\crefname{assumption}{Assumption}{Assumptions}
	\crefname{construction}{Construction}{Constructions}
	\crefname{corollary}{Corollary}{Corollaries}
	\crefname{conjecture}{Conjecture}{Conjectures}
	\crefname{definition}{Definition}{Definitions}
	\crefname{example}{Example}{Examples}
	\crefname{experiment}{Experiment}{Experiments}
	\crefname{counterexample}{Counterexample}{Counterexamples}
	\crefname{lemma}{Lemma}{Lemmata}
	\crefname{observation}{Observation}{Observations}
	\crefname{proposition}{Proposition}{Propositions}
	\crefname{remark}{Remark}{Remarks}
	\crefname{claim}{Claim}{Claims}
	\crefname{fact}{Fact}{Facts}
	\crefname{note}{Note}{Notes}

\ifnum\llncs=1
 \crefname{appendix}{App.}{Appendices}
 \crefname{section}{Sec.}{Sections}
\else
\fi

\ifnum\llncs=1
\renewcommand*{\backref}[1]{}
\else
	\renewcommand*{\backref}[1]{(Cited on page~#1.)}
	\ifnum\notxfont=1
	\else
		\usepackage{newtxtext}
	\fi
\fi
\ifnum\authnote=0  
\newcommand{\mor}[1]{}
\newcommand{\minki}[1]{}
\newcommand{\takashi}[1]{}

\else
\newcommand{\shira}[1]{$\ll$\textsf{\color{purple} Yuki: { #1}}$\gg$}
\newcommand{\taiga}[1]{$\ll$\textsf{\color{red} Taiga: { #1}}$\gg$}

\fi

\newcommand{\Tr}{\mathrm{Tr}}

\newcommand{\SD}{\mathsf{SD}} 

\newcommand{\StateGen}{\mathsf{StateGen}}

\newcommand{\puzz}{\mathsf{puzz}}
\newcommand{\ans}{\mathsf{ans}}

\newcommand{\Samp}{\algo{Samp}}

\newcommand{\cA}{\mathcal{A}}

\newcommand{\cC}{\mathcal{C}}
\newcommand{\cD}{\mathcal{D}}
\newcommand{\cE}{\mathcal{E}}

\newcommand{\cG}{\mathcal{G}}

\newcommand{\cI}{\mathcal{I}}

\newcommand{\cL}{\mathcal{L}}
\newcommand{\cM}{\mathcal{M}}
\newcommand{\cN}{\mathcal{N}}

\newcommand{\cP}{\mathcal{P}}
\newcommand{\cQ}{\mathcal{Q}}
\newcommand{\cR}{\mathcal{R}}

\newcommand{\cT}{\mathcal{T}}

\def\makeuppercase#1{
\expandafter\newcommand\csname tl#1\endcsname{\widetilde{#1}}
}

\def\makelowercase#1{
\expandafter\newcommand\csname tl#1\endcsname{\widetilde{#1}}
}

\newcommand{\N}{\mathbb{N}}

\newcommand{\regA}{\textcolor{gray}{\mathsf{A}}}
\newcommand{\regB}{\textcolor{gray}{\mathsf{B}}}
\newcommand{\regC}{\textcolor{gray}{\mathsf{C}}}

\newcommand{\regR}{\textcolor{gray}{\mathsf{R}}}

\newcommand{\regX}{\textcolor{gray}{\mathsf{X}}}

\newcommand{\secp}{\lambda}

\newcommand{\param}{\mathsf{param}}

\newcommand*{\algo}[1]{\ensuremath{\mathsf{#1}}}

\newenvironment{boxfig}[2]{\begin{figure}[#1]\fbox{\begin{minipage}{0.97\linewidth}
                        \vspace{0.2em}
                        \makebox[0.025\linewidth]{}
                        \begin{minipage}{0.95\linewidth}
            {{
                        #2 }}
                        \end{minipage}
                        \vspace{0.2em}
                        \end{minipage}}}{\end{figure}}

\newcommand{\bit}{\{0,1\}}

\newcommand{\KeyGen}{\algo{KeyGen}}

\newcommand{\Ver}{\algo{Ver}}

\newcommand{\TD}{\algo{TD}}

\newcommand{\negl}{{\mathsf{negl}}}

\newcommand{\poly}{{\mathrm{poly}}}

\makeatletter
\DeclareRobustCommand
  \myvdots{\vbox{\baselineskip4\p@ \lineskiplimit\z@
    \hbox{.}\hbox{.}\hbox{.}}}
\makeatother
 
\title{Universal Inductive Inference of Quantum States}

\ifnum\anonymous=1
\ifnum\llncs=1
\author{\empty}\institute{\empty}
\else
\author{}
\fi
\else
\ifnum\llncs=1
\author{
Taiga Hiroka\inst{1} \and Min-Hsiu Hsieh\inst{1} \and Yuki Shirakawa\inst{2}
}
\institute{
 Foxconn, Taipei, Taiwan \and Yukawa Institute for Theoretical Physics, Kyoto University, Kyoto, Japan
}
\else
\author[1]{Taiga Hiroka%
\thanks{Email: \texttt{taiga.hirooka@foxconn.com}}}

\author[1]{Min-Hsiu Hsieh%
\thanks{Email: \texttt{min-hsiu.hsieh@foxconn.com}}}
\author[2,1]{Yuki Shirakawa%
\thanks{Email:
\texttt{yuki.shirakawa@yukawa.kyoto-u.ac.jp}}}
\affil[1]{{\small Hon Hai (Foxconn) Research Institute, Taipei, Taiwan}}
\affil[2]{{\small Yukawa Institute for Theoretical Physics, Kyoto University, Kyoto, Japan}}
\fi 
\fi

\date{}

\begin{document}

\maketitle

\begin{abstract}
Solomonoff's universal inductive inference [Inf. Control. 1964] is one of the most general frameworks for learning from sequential data and predicting future observations.
It provides prediction guarantees for any computable stochastic source. We ask whether an analogous universal theory of induction can be developed for quantum systems. 
To this end, we introduce \emph{universal quantum inductive inference}, a framework for learning and predicting quantum sources that may exhibit arbitrary correlations and entanglement across time. 
Given the outcomes of past measurements and the corresponding post-measurement quantum systems, the learner produces a joint state that predicts the next measurement outcome and post-measurement system while preserving correlations with the past systems.
We model a quantum source as a multipartite quantum state whose description is generated by an unknown $s$-bit program within a given time bound. 
The learner is required to produce a state that is $\epsilon$-close in trace distance to the target with probability at least $1-\delta$.
We establish an information-theoretic inference algorithm with round complexity $O(s\epsilon^{-2}\delta^{-1})$, with no additional dependence on the number of qubits received in each round or the time required to generate the source. 
We further prove a matching lower bound in the relevant parameter regime, even for classical sources.
As a second result, we construct a non-i.i.d. state tomography algorithm that outputs a classical description of the conditional state of the next system given past measurement outcomes. 
Our information-theoretic tomography algorithm achieves round complexity $O(sn^2\epsilon^{-4}\delta^{-2})$, where $n$ is the number of qubits received in each round. 
Previous works [Fawzi, Kueng, Markham, and Oufkir, Nat. Commun. 2024; Zambrano, 2026] have also established tomography guarantees in non-i.i.d. settings, but their target states are defined through averaging and therefore do not, in general, provide guarantees for predicting the next system from the preceding systems in their temporal order. 
Our results provide such prediction guarantees even in the presence of arbitrary correlations and entanglement across time.
Finally, we investigate the computational hardness of universal quantum inductive inference under quantum cryptographic assumptions. 
We show that the existence of (infinitely-often) one-way puzzles [Khurana and Tomer, STOC 2024] implies average-case hardness of quantum inductive inference with polynomial round complexity. 
Moreover, for a restricted variant in which the learner observes only classical measurement outcomes and predicts the conditional distribution of the next outcome, we obtain a complete characterization: average-case hardness is equivalent to the existence of infinitely-often one-way puzzles.
\end{abstract}

\newpage

 \setcounter{tocdepth}{2}
 \tableofcontents
 \newpage

\section{Introduction}\label{sec:introduction}

Learning from observations often involves inferring what comes next.
When observations are independent and identically distributed,
this reduces to learning a fixed distribution.
In many realistic settings, however, the source may change over time, and its outputs may depend on previous events.
A general theory of inference should therefore accommodate behavior whose structure is initially unknown.

In the classical setting, Solomonoff induction provides a general framework for this task~\cite{Sol60,Sol64,Sol64II}.
The learner receives a prefix of a sequence generated by an unknown source and infers the conditional distribution of the next symbol.
The source need not be i.i.d.\ or stationary.
For every computable source, Solomonoff induction eventually infers its future outputs accurately~\cite{Sol78}.
In particular, the learner need not know in advance whether the source is periodic or has some other structure.
This framework has in turn inspired computationally feasible approximations for general reinforcement learning~\cite{VNHUS11}. Solomonoff induction thus provides a general benchmark for inference
from classical sequential data.

We ask whether such inference is possible for quantum sources. Whereas a classical history can be read without being altered, obtaining classical outcomes from quantum systems can disturb their states and change their correlations with future systems.
Moreover, even after conditioning on these outcomes, the remaining quantum systems may still be entangled with future systems.
These features make quantum inductive inference fundamentally different from its classical counterpart: applying classical inductive inference only to measurement outcomes is insufficient to predict future quantum systems and their correlations with the past. 

Our task can be viewed as a form of non-i.i.d. quantum tomography. One might therefore expect existing results on non-i.i.d. tomography \cite{FKMO24,Zam26b} and certification \cite{CR12,ZSK26,NZ26} to provide the desired inference guarantee. However, these frameworks share a common limitation. Their guarantees concern averaged states or properties, and do not require learning the temporal structure needed to predict and recover the subsequent outputs.

To illustrate this limitation, consider the bit strings $010101\ldots01$ and $101010\ldots10$ of the same even length, encoded in the computational basis. Their time-averaged single-system states are both $I/2$, and permutation averaging makes their joint states identical. The tomography procedures of \cite{FKMO24,Zam26b} yield estimates close to $I/2$ for either input with high probability. Yet the two strings differ at every position. Such estimates describe the frequencies of the bits but do not determine which bit comes next. We find it unsatisfactory to regard learning these frequencies alone as successful learning of the source. A satisfactory notion of learning should require the learner to exploit simple temporal regularities, such as the alternating pattern in this example, to predict and recover subsequent outputs. Solomonoff induction provides such predictive guarantees for computable classical sources, with predictions converging to the true conditional probabilities of subsequent symbols. We seek an analogous guarantee for quantum sources. Specifically, we ask the learner to infer the conditional state of the next system from the preceding systems in their original order, including for probabilistic quantum sources with arbitrary correlations and entanglement.

We therefore study inference that retains the temporal correlations
of a quantum source.
The learner receives the classical outcomes of past measurements
together with the corresponding post-measurement quantum systems.
Its inferred state must reproduce the next measurement outcome and
the next post-measurement system, conditioned on the observed history.
We also require the inferred state to reproduce their correlations
with the past systems.
These correlations are necessary for inferring the statistics of
joint measurements involving both past and future systems.

It is not clear whether this task is possible even with unlimited
computation.
The learner receives only a single copy of past systems and no description of the target source. 
Our main question is the following:
\begin{quote}
\itshape
Can we formulate quantum inductive inference that is achievable
information-theoretically, with the required number of rounds
controlled by the description length of the source?
\end{quote}

In this work, we introduce such a formulation and give an affirmative
answer.
We consider quantum sources generated by an unknown process with
a bounded classical description length and a bounded preparation time.
We show that inference is possible even when the generated systems
are correlated or entangled.
Information-theoretic feasibility alone does not tell us whether
an efficient learner can infer every efficiently generatable source.
We therefore also study the computational hardness of this task
and its relation to quantum cryptography.
We describe the definition and our results below.

\subsection{Our Results}
We first formalize \emph{universal quantum inductive inference} (See also \cref{fig:quantum-induction}).
We model an unknown quantum source as a quantum state whose preparation circuit is generated by a resource-bounded Turing machine.
More precisely, a quantum source $\rho_{\regR_1...\regR_m}$ is an $mn$-qubit state on $m$ registers $(\regR_1,...,\regR_m)$, where each register $\regR_i$ consists of $n$ qubits, and the circuit description of $\rho_{\regR_1,...,\regR_m}$ is generated within time $t$ by a Turing machine $T$ with description length at most $s$.
We do not assume specific structures of $\rho_{\regR_1...\regR_m}$ and the registers may be arbitrarily entangled in general.

\begin{figure}[H]
  \centering
  \resizebox{0.85\linewidth}{!}{%
    \begin{tikzpicture}[
  x=1cm,y=1cm,font=\small,
  box/.style={draw=black!60,rounded corners=2pt,align=center,inner sep=8pt},
  flow/.style={-{Stealth[length=2mm]},semithick},
  note/.style={font=\footnotesize,text=black!65,align=center}
]
\node[font=\bfseries] at (2.3,3.15) {Inputs};
\node[box,text width=4.2cm,minimum height=1.15cm] (classical) at (2.3,1.65)
  {Past measurement results\\[3pt]$\mathcal I_1,x_1,\ldots,\mathcal I_{i-1},x_{i-1}$\\[6pt]
   Next measurement\\[3pt]$\cI_i$};
\node[box,text width=4.2cm,minimum height=1.1cm] (quantum) at (2.3,-0.9)
  {Post-measurement state\\on registers $(\regR_1,...,\regR_{i-1})$\\[4pt]$\rho_{\regR_1...\regR_{i-1}|h_{<i}}$};
\node[box,fill=black!6,text width=2cm,minimum height=1.25cm] (learner) at (7.0,0.3)
  {Learner $\mathcal L$};
\draw[flow] (classical.east) -- (5.25,1.65) |- (learner.160);
\draw[flow] (quantum.east) -- (5.25,-0.9) |- (learner.200);
\node[note] at (7.0,-0.8) {Parameters $n,s,t,\epsilon,\delta$};

\node[font=\bfseries] at (12.4,3.15) {Output};
\draw[rounded corners=2pt,black!60] (9.1,-1.25) rectangle (15.7,1.8);
\node at (12.4,1.2) {Joint state $\sigma_{\regR_1...\regR_i\regX_i}$};
\draw[black!25] (9.1,0.6) -- (15.7,0.6);
\draw[black!25] (11.3,-1.25) -- (11.3,0.6);
\draw[black!25] (13.5,-1.25) -- (13.5,0.6);
\node[align=center] at (10.2,0.1) {$\regR_1,...,\regR_{i-1}$};
\node[align=center] at (12.4,0.1) {$\regR_i$};
\node[align=center] at (14.6,0.1) {$\regX_i$};
\node[note] at (10.2,-0.65) {Given quantum\\registers};
\node[note] at (12.4,-0.65) {Next quantum\\register};
\node[note] at (14.6,-0.65) {Next classical\\outcome};
\draw[flow] (learner.east) -- (9.1,0.3);
\node[note] at (12.4,-1.8) {Conditioned on the observations $x_{<i}$};
\end{tikzpicture}%
  }
  \caption{
    Inputs and output of the learner $\mathcal{L}$ at round $i$.
    The instruments $\mathcal{I}_1,\ldots,\mathcal{I}_{i-1}$ and their outcomes $x_1,\ldots,x_{i-1}$ form the classical history $h_{<i}=(\cI_1,x_1,...,\cI_{i-1},x_{i-1})$.
    The learner receives $h_{<i}$, the description of the target measurement $\mathcal{I}_i$, and a single copy of the conditional post-measurement state $\rho_{\regR_1...\regR_{i-1}\mid h_{<i}}$.
    It outputs a joint state $\sigma$ on $\regR_1...\regR_i\regX_i$, where $\regX_i$ stores the next measurement outcome and $\regR_i$ is the next post-measurement quantum system.
    The output must approximate the corresponding joint state of the source conditioned on $h_{<i}$, including its correlations with the past systems.
  }
  \label{fig:quantum-induction}
\end{figure}
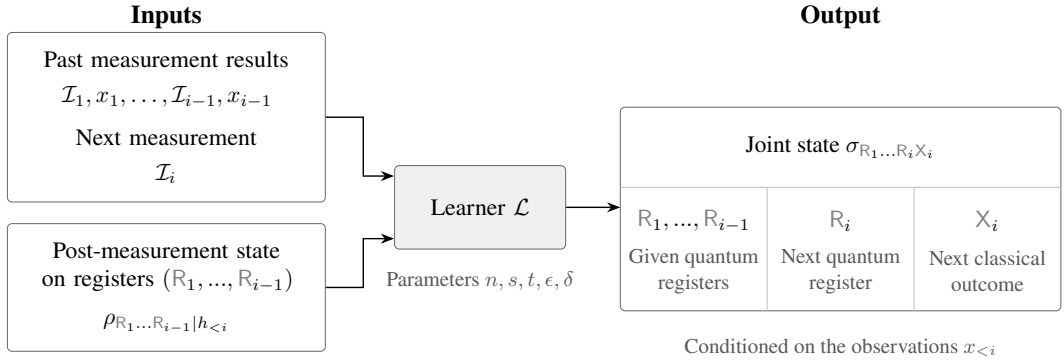

At round $i$, the preceding registers $\regR_1,\ldots,\regR_{i-1}$ have been measured using quantum instruments $\mathcal{I}_1,\ldots,\mathcal{I}_{i-1}$, producing classical outcomes $x_1,\ldots,x_{i-1}$.
Each instrument specifies both the classical outcome statistics and the corresponding post-measurement states.
We denote the resulting classical history by $h_{<i}:=(\mathcal{I}_1,x_1,\ldots,\mathcal{I}_{i-1},x_{i-1})$
and the remaining post-measurement state on $(\regR_1,...,\regR_{i-1})$, conditioned on this history, by
$\rho_{\regR_1...\regR_{i-1}\mid h_{<i}}$.
The learner $\mathcal{L}$ receives $h_{<i}$, the description of the next instrument $\mathcal{I}_i$, and a single copy of $\rho_{\regR_1...\regR_{i-1}\mid h_{<i}}$, together with the parameters $n,s,t,\epsilon,\delta$.

The learner outputs a joint state $\sigma$ on $\regR_1...\regR_i\regX_i$, aiming to reproduce the next classical outcome in $\regX_i$, the next post-measurement system $\regR_i$, and their correlations with the past systems $\regR_1...\regR_{i-1}$.
The target $\rho_{\regR_1...\regR_i\regX_i\mid h_{<i}}$ is the joint state obtained by applying $\mathcal{I}_i$ to the actual next register of the source, conditioned on $h_{<i}$.
We require $\TD(\sigma,\rho_{\regR_1...\regR_i\regX_i\mid h_{<i}})\leq\epsilon$ with probability at least $1-\delta$ over a uniformly chosen $i\in[m]$ and the preceding measurement outcomes, where $\TD$ denotes trace distance.
We emphasize that, except for the parameters, the learner $\cL$ is given no description of the quantum source $\rho_{\regR_1...\regR_m}$ or of the Turing machine that generates it. 
Moreover, we require the existence of a single universal learner that makes accurate predictions for any such quantum source and any sequence of measurements.

Our formulation accommodates a broad range of models for quantum sequential inference.
For example, one may describe a source as generating one register at a time while retaining an internal quantum memory.
Each subsequent register is generated together with an updated memory, possibly using classical outcomes obtained during earlier steps of the source's computation.
Such a sequential process can be represented by a single circuit generating its joint output state and is included whenever this circuit can be generated within time $t$ by a Turing machine of description length at most $s$.
Moreover, when the source is classical and all instruments measure in the computational basis, the task reduces to sampling the next symbol from its conditional distribution given the observed history.
Thus, although there are several possible ways to formulate quantum inductive inference, our definition accommodates quantum memory and correlations across time while recovering the sampling formulation of classical Solomonoff induction~\cite{Sol60,Sol64,Sol64II} as a special case.

Our first result shows that universal quantum inductive inference is information-theoretically possible as follows:
\begin{theorem}[Universal quantum inductive inference, Restated in \cref{thm:info_alg}]
\label{thm:main}
There exists an algorithm that solves universal quantum
inductive inference with round complexity
\[
    r(s,\epsilon^{-1},\delta^{-1})
    = O\!\left(\frac{s}{\epsilon^2\delta}\right).
\]
\end{theorem}

Our round complexity matches the classical bound $O(s/(\epsilon^2\delta))$ of Solomonoff induction~\cite{Sol78},
as formulated for time-bounded sources by Hirahara and Nanashima~\cite{HN26}.
In their setting, the learner receives a prefix generated by an unknown randomized Turing machine of description length at most $s$ and running time at most $t$, and must sample from a distribution approximating the conditional distribution of the next symbol.
Our theorem recovers this round complexity without additional dependence on the register size $n$ or the preparation time $t$.
Hirahara and Nanashima~\cite[Proposition~B.1]{HN26} also prove a classical lower bound of $\Omega(s/\epsilon^2+s/\delta)$.
In \cref{sec:lower_bound}, we improve this lower bound to $\Omega(s/(\epsilon^2\delta))$, establishing optimality up to constant factors for both the classical and quantum tasks.

\paragraph{Non-i.i.d.\ state tomography.}
Our algorithm that solves universal inductive inference in \cref{thm:main} produces a single-copy of an inferred quantum state, but does not provide a classical description of it.
We next ask whether one can obtain a classical description of the target state from a non-i.i.d. quantum source.
We call this task \emph{non-i.i.d. state tomography} and formulate it as follows.
A quantum source is an $mn$-qubit state $\rho_{\regR_1...\regR_m}$ whose circuit description is generated by a resource-bounded Turing machine as above.
Again, the description of $\rho_{\regR_1...\regR_m}$ or of the underlying Turing machine is unknown.
For a randomly chosen round $i\in[m]$, a learner $\cL$ receives a prefix state $\rho_{\regR_1...\regR_{i-1}}$ on the first $i-1$ registers, measures it by a POVM $M$, and obtains a classical outcome $x$.
The task of $\cL$ is to produce a classical description of an $n$-qubit state that is $\epsilon$-close to the true $i$th state $\rho_{\regR_i|x}$ conditioned on the measurement outcome $x$, with probability at least $1-\delta$.
We allow $\cL$ to choose its own measurement $M$ on the given registers.
Even with this freedom of the measurement, information-theoretic feasibility is not immediate: the learner must obtain a classical description of the target state $\rho_{\regR_i|x}$ from a single-copy of the prefix state $\rho_{\regR_1...\regR_{i-1}}$.
Our second result establishes the feasibility of non-i.i.d. state tomography.

\begin{theorem}[Non-i.i.d. state tomography, Restated in \cref{thm:state_tomography_improve}]
\label{thm:main_tom}
There exists an algorithm that solves non-i.i.d.\ state tomography with round complexity
\[
    r(n,s,\epsilon^{-1},\delta^{-1})
    =
    O\!\left(
        \frac{sn^2}{\epsilon^4\delta^2}
    \right).
\]
The learner chooses and performs a measurement on the preceding registers $\regR_1,\ldots,\regR_{i-1}$ and outputs a classical description of a state approximating the conditional state of $\regR_i$.
The trace-distance error is at most $\epsilon$ with probability at least $1-\delta$ over a uniformly chosen round and the learner's measurement outcomes.
\end{theorem}



Let us compare our result with previous work on non-i.i.d. quantum tomography \cite{FKMO24,Zam26b}. There are two main differences. First, their learning targets involve averaging over registers. As discussed in the introduction, such averaging can erase information about correlations between successive registers. Our target is the conditional state of the next register. We produce a classical description of this state using only the preceding registers, while preserving their original order. Second, their general full-state tomography bounds scale exponentially with the number of qubits \(n\), whereas our round complexity is polynomial in the source description length \(s\), the register size \(n\), and \(\epsilon^{-1},\delta^{-1}\). Thus, sources with generative descriptions of length \(s=\operatorname{poly}(n)\) can be learned with polynomially many registers, even in the presence of arbitrary correlations and entanglement. We note that the dependence on \(s\) is essential in our setting. The lower bound in \cref{sec:lower_bound} also applies to this tomography setting and implies an \(\Omega(s)\) round complexity even for single-qubit classical sources and fixed accuracy and confidence parameters.

Our result can also be viewed as a fully non-i.i.d. analogue of quantum hypothesis selection \cite{CL21,BO24}. In the usual i.i.d. setting, one is given copies of an unknown quantum state and a finite set of candidate states, and seeks a candidate close to the unknown state. Hypothesis selection has also been studied for independent but non-identically distributed inputs \cite{FQR24}. In our setting, the hypotheses describe entire quantum sources, and the learner produces an approximation to the next conditional state from the preceding registers. The source may exhibit arbitrary correlations and entanglement, and its output registers need not be identically distributed. To the best of our knowledge, this is the first such sequential tomography guarantee with sample complexity polynomial in \(s\) and \(n\), extending existing quantum hypothesis selection to fully non-i.i.d. sources.

\paragraph{Computational Complexity.}
Our results so far establish information-theoretic inference guarantees,
but do not guarantee computational efficiency.
It is natural to ask whether we can obtain an efficient learner
inferring the outputs of every efficiently generatable source.
In the classical setting, the existence of one-way functions rules out
such a learner, even with polynomially many observations and constant
error and failure parameters~\cite{NR06,HN23}.
We ask whether computational hardness of quantum inference follows
from more general quantum cryptographic assumptions.

We show that infinitely-often one-way puzzles
(OWPuzzs)~\cite{STOC:KhuTom24} suffice.
Informally, these puzzles have an efficient quantum sampler that
generates a classical puzzle together with a valid answer.
For every efficient quantum algorithm, there are infinitely many
security parameters on which its probability of producing a valid
answer given only the puzzle is negligible.
The verification procedure need not be efficient.
Under this assumption, we show that no efficient learner can achieve
the information-theoretic round complexity established above,
even in the average-case setting.

\begin{theorem}[Hardness of universal quantum inductive inference, Restated in \cref{thm:owpuzz_to_hard}]
\label{thm:main_hard}
Suppose that infinitely-often OWPuzzs exist.
Then no quantum polynomial-time algorithm solves average-case
universal quantum inductive inference with polynomial round complexity.
\end{theorem}

The theorem above shows that the existence of infinitely-often
OWPuzzs is a sufficient condition for average-case hardness of
universal quantum inductive inference at this round complexity.
A natural question is whether this condition is also necessary.
We leave this converse direction as an open problem.

Although the converse remains open for general quantum inference,
we obtain a complete cryptographic characterization when the learner
observes and infers only classical measurement outcomes.
Consider a quantum source whose successive output registers are measured
using specified POVMs.
The learner receives the descriptions of the past measurements and
the next measurement, together with the past measurement outcomes.
It must sample from a distribution approximating the conditional
distribution of the next outcome given the observed history.
The learner may perform quantum computation, but has no access to
the source's quantum registers.
We call this task \emph{POVM-based universal quantum inductive inference}.

In the average-case formulation, the source and measurement descriptions
are sampled by a quantum polynomial-time algorithm.
As before, we require error at most $\epsilon$ with probability
at least $1-\delta$ over the sampled instance, a uniformly chosen round,
and the preceding measurement outcomes.
Here, error is measured in statistical distance between the learner's
output distribution and the true conditional distribution.
We show that average-case computational hardness of this task is equivalent to
the existence of infinitely-often OWPuzzs.

\begin{theorem}[Hardness of POVM-based universal quantum inductive inference, Restated in \cref{thm:owpuzz_povmhard}]
\label{thm:main_povm}
There exists no quantum polynomial-time algorithm that solves
average-case POVM-based universal quantum inductive inference
with round complexity
\[
    O\left(
        \frac{s(\lambda)+g(\secp)}
        {\epsilon(\lambda)^2\delta(\lambda)}
    \right)
\]
for any superconstant function $g(\secp)=\omega(1)$, if and only if infinitely-often OWPuzzs exist.
\end{theorem}

\subsection{Related Works}

\paragraph{Non-i.i.d.\ quantum state learning.}
Several works study quantum state learning without the i.i.d.\ assumption~\cite{CR12,BH17,FQR24,FKMO24,Zam26a,Zam26b}.
Among these, the framework of Fawzi, Kueng, Markham, and Oufkir~\cite{FKMO24} and the two works of Zambrano~\cite{Zam26a,Zam26b} address tasks particularly close to our tomography task, and we discuss them in more detail below.
Fawzi, Kueng, Markham, and Oufkir introduce a general framework for extending learning algorithms designed for i.i.d.\ inputs to arbitrary, possibly entangled inputs.
In their framework, the input registers are randomly permuted, and the learner uses the training registers to infer properties of a remaining test register.
The learner's output is evaluated against the test register's state conditioned on that output and calibration information.
Zambrano studies a different setting in which each prepared state may depend on the previous experimental history.
He gives algorithms for reconstructing the time-averaged state~\cite{Zam26a} and estimating its observable expectation values~\cite{Zam26b}, with sample complexities comparable to the corresponding i.i.d.\ guarantees.
These results do not require a bound on the source description length, but their inference targets differ from ours.
In our setting, the learner uses only the preceding registers to infer the conditional state of the next register in the original sequence.
We exploit a bound on the source description length to obtain this guarantee.
Fanizza, Quek, and Rosati~\cite{FQR24} also consider a related learning task, but assume that the input registers are independent, although not necessarily identically distributed.
Given a collection of candidate state sequences, their learner selects a candidate whose average per-register trace distance from the input states is approximately optimal.
Their learner may measure the entire input sequence, whereas ours infers the state of an unobserved next register using only its predecessors and allows arbitrary correlations and entanglement.

\paragraph{Non-i.i.d.\ quantum state verification.}
Quantum state verification and certification without the i.i.d.\ assumption have been studied extensively~\cite{CR12,HM15,MTH17,TM18,TMMMF19,ZH19a,ZH19b,ZH19c,MK20,LZH23,FKMO24,AGMOC25}.
For broad classes of pure targets, polynomially many samples suffice even when the supplied registers are arbitrarily correlated or entangled~\cite{MTH17,TM18,TMMMF19}.
These protocols certify high fidelity with a known pure state.
These verification protocols rely on the purity of the target and do not apply to general mixed targets.

For mixed targets, Fawzi, Kueng, Markham, and Oufkir~\cite{FKMO24} give a test of trace-distance proximity to the maximally mixed state within their framework for correlated inputs.
Their soundness guarantee concerns a remaining test register, conditioned on the learner's output and calibration information, after randomly permuting the input registers.
Their sample bound is polynomial in the Hilbert-space dimension and hence exponential in the number of qubits.
De Palma, Fanizza, Mowry, and O'Donnell~\cite{PFMO25} instead study independent but not necessarily identically distributed inputs.
Given the product state $\rho_1\otimes\cdots\otimes\rho_N$, they test whether the average state equals a known mixed target or is more than $\epsilon$ away in trace distance, using $O(2^n/\epsilon^2)$ samples for $n$-qubit states and constant success probability.
Their result retains the independence assumption and therefore does not cover general correlated or entangled inputs.

Cavalar et al.~\cite{CGGHHM25} obtain polynomial sample complexity for mixed-state verification against time-bounded uniform quantum adversaries supplying arbitrarily correlated registers.
They verify whether the average marginal state is close to a known,
efficiently generatable target.
Our task instead infers an unknown next state conditioned on the observed history, including its correlations with retained past systems.
This conditioning is essential, and it is observed that learning unconditional marginals from arbitrarily correlated samples is information-theoretically impossible in general, even for classical sources~\cite{CGGHHM25}.

\paragraph{Computational complexity of learning and inference.}
The relationship between learning and cryptography has been studied extensively, with cryptographic assumptions providing barriers to efficient learning~\cite{Val84} and the nonexistence of one-way functions enabling learning algorithms~\cite{IL90,BFKL93,NR06}.
Naor and Rothblum~\cite{NR06} study the learning of adaptively changing distributions and characterize efficient learnability in their model by the nonexistence of one-way functions (OWFs).
Hirahara and Nanashima~\cite{HN23} develop a unified framework based on inductive inference for obtaining learning algorithms under the nonexistence of OWFs.
Their results include a characterization of OWFs by the average-case hardness of improper distribution learning.
More recently, Hirahara and Nanashima~\cite{HN26} give an efficient universal inductive inference algorithm under an average-case tractability assumption on time-bounded Kolmogorov complexity.
Their algorithm infers future observations from a sequence generated by an unknown time-bounded randomized machine.
This provides a computational counterpart to classical universal inductive inference and is closely related to our goal of inferring the states of quantum systems from past observations.

In the quantum setting, Hiroka, Hsieh, and Morimae~\cite{HHM25} characterize one-way puzzles by the average-case hardness of proper learning of quantumly samplable classical distributions.
Analogous connections between quantum state learning and one-way state generators have been established for pure states~\cite{HH24,FGSY25}, and for mixed states using inefficiently verifiable one-way state generators~\cite{CL26}.
These learning tasks use i.i.d.\ samples of a fixed unknown distribution or state.
Related works connect quantum cryptography to concrete learning and decoding assumptions~\cite{BHHP24,PQS25}, and to the average-case hardness of probability estimation and Kolmogorov complexity~\cite{CGGH25,HM25}.

Recent works~\cite{QRZ25,STOC:KhuTom25} study classical-to-quantum extrapolation: given the outcome of measuring one register of a known bipartite pure-state preparation, the goal is to prepare the conditional state of the other register.
They show that hardness of this task implies quantum commitments.
They also introduce a fully quantum extrapolation task that follows from quantum commitments.
Their tasks concern a known preparation procedure, whereas our learner performs inference from a sequence generated by an unknown source.

Qian and Zhandry~\cite{QZ26} study a related setting in which an eavesdropper observes classical communication between quantum parties.
They show that, after sufficiently many rounds, the eavesdropper can prepare a quantum state that enables impersonation in a subsequent interaction.
Their task can be viewed as a form of inductive inference from classical transcripts to quantum states.
In contrast, our learner receives the retained post-measurement quantum systems, together with the classical measurement history, and predicts the next outcome and post-measurement state, including correlations with the past systems.
Thus, the two works consider different models of quantum inductive inference.

Cavalar et al.~\cite{CGGHHM25} give sufficient conditions for efficient verification of quantumly samplable distributions and efficiently generatable mixed states under the nonexistence of infinitely-often one-way puzzles and weak non-uniform infinitely-often EFI pairs, respectively.
Their protocols allow correlated inputs from time-bounded adversaries, but verify proximity to a specified target.
Our task instead requires inferring an unknown next state from the preceding registers.

\paragraph{Quantum algorithmic information theory.}
Kolmogorov complexity measures the length of a shortest program describing an object~\cite{Sol64,Kol65,Cha66}.
Universal algorithmic probability gives a probabilistic counterpart to this notion: it combines the contributions of programs, giving greater weight to shorter programs.
Solomonoff induction uses a related universal prior to infer future observations, with guarantees governed by the description length of the underlying computable source~\cite{Sol64,Sol78}.
Quantum extensions of algorithmic information theory have been studied since the work of Svozil~\cite{Svo96}, with notions of description complexity based on classical descriptions~\cite{Vit01,MB05} and quantum programs~\cite{BDL01}.
Gács~\cite{Gac01} develops a quantum analogue of universal probability and uses it to define quantum algorithmic entropy. Cavalar et al.~\cite{CCCGHJL26} characterize the existence of EFI pairs by the average-case hardness of estimating a smoothed version of Gács's algorithmic entropy from a single copy of a pure state drawn from a non-uniformly samplable distribution. Perrier~\cite{Perrier26} proposes a quantum analogue of Solomonoff induction based on a description-length-weighted mixture of quantum environments for predicting measurement outcomes. The paper states a convergence bound for its posterior environment state, but explicitly leaves a complete proof open, and their formulation is different from ours.
Our inference algorithm uses a time-bounded variant of Gács's universal construction to predict a conditional joint quantum state.

\section{Technical Overview}
\label{sec:technical_overview}

\subsection{Universal Quantum Inductive Inference}
We first describe an information-theoretic algorithm that achieves the upper bound in \cref{thm:main}.
In this overview, for simplicity, we consider the special case of universal quantum inductive inference in which the target instrument is $\cI_i=\mathrm{id}_{\regR_i}$ for all $i\in[m]$, where $m\in\N$ denotes the total number of rounds.
In this case, the task reduces to the following: for a randomly chosen $i\in[m]$, given the reduced state $\rho_{\regR_1\cdots\regR_{i-1}}$ on the first $i-1$ registers $(\regR_1,...,\regR_{i-1})$, generate a state that is $\epsilon$-close to $\rho_{\regR_1...\regR_i}$ in trace distance.
Here, the circuit description of the target state $\rho_{\regR_1...\regR_m}$ is generated by a Turing machine $T$ described by $s$ bits within $t$ steps.

The starting point of our construction is the Petz recovery map \cite{Petz86,Petz88}.
Because our task is to construct the state $\rho_{\regR_1...\regR_i}$ from its reduced state $\rho_{\regR_1...\regR_{i-1}}$ on the first $i-1$ registers, it can be viewed as a state-recovery problem for the channel $\Tr_{\regR_i}$ that traces out the $i$th register.
Indeed, if sufficient information about the target state $\rho_{\regR_1...\regR_i}$ is available, the Petz recovery map can recover it exactly.
More precisely, if we have a reference state $\sigma_{\regR_1...\regR_i}$ that satisfies $D(\rho_{\regR_1...\regR_i}\|\sigma_{\regR_1...\regR_i})=D(\rho_{\regR_1...\regR_{i-1}}\|\sigma_{\regR_1...\regR_{i-1}})$\footnote{
For quantum states $\rho$ and $\sigma$, $D(\rho\|\sigma)$ is the relative entropy defined by $D(\rho\|\sigma)\coloneqq \Tr \rho(\log\rho-\log\sigma)$.
}, then the Petz recovery map $\cP_{\sigma,\Tr_{\regR_i}}$ with respect to $\sigma_{\regR_1...\regR_i}$ satisfies $\cP_{\sigma,\Tr_{\regR_i}}(\rho_{\regR_1...\regR_{i-1}})=\rho_{\regR_1...\regR_i}$.
However, the main obstacle in our setting is that the source program $T$ generating the target state is unknown to the learner. 
It is therefore highly nontrivial to prepare a reference state for which the corresponding Petz recovery map successfully recovers the target state.

In this work, we overcome this obstacle and establish the achievability of universal quantum inductive inference using the notion of a \emph{time-bounded universal density matrix}\footnote{
    Gács~\cite{Gac01} introduced the universal (semi)density matrix. Our notion can be viewed as its time-bounded variant.
} and its domination property.
Roughly speaking, a $t$-time-bounded universal density matrix $\mu^t$ is the uniform mixture of quantum states whose preparation circuits are generated by a universal Turing machine from $t$-bit programs within $t$ steps.
This is a quantum-state analogue of the (time-bounded) universal distribution well known in classical algorithmic information theory and meta-complexity~\cite{Sol64,Kol65,Cha66}.
Analogously to the universal distribution, it satisfies the following domination property: there exists a polynomial $p$ such that, for any state $\rho_{\regR_1...\regR_m}$ whose preparation circuit is generated by a Turing machine described by $s$ bits within $t$ steps,
\begin{align}
    \mu^{p(s,t)}_{\regR_1...\regR_m} \ge 2^{-O(s)} \rho_{\regR_1...\regR_m}.
    \label{eq:overview_domination}
\end{align}
In particular, this matrix inequality holds for our target state and implies that $D(\rho_{\regR_1...\regR_m}\|\mu^{p(s,t)}_{\regR_1...\regR_m})\le O(s)$.
We use the universal density matrix $\mu^{p(s,t)}_{\regR_1...\regR_m}$ as the reference state and apply the rotated Petz recovery map $\cR_{\mu,\Tr_{\regR_i}}$ \cite{Wil15} to the given prefix state $\rho_{\regR_1...\regR_{i-1}}$.
By prior work \cite{JRSWW18}, we observe that $\cR_{\mu,\Tr_{\regR_i}}$ satisfies
\begin{align}
    D(\rho_{\regR_1...\regR_i}\|\mu_{\regR_1..\regR_i})- D(\rho_{\regR_1...\regR_{i-1}}\|\mu_{\regR_1..\regR_{i-1}}) \ge -2\log F\left(\rho_{\regR_1...\regR_i}, \cR_{\mu,\Tr_{\regR_i}}(\rho_{\regR_1...\regR_{i-1}})\right),
\end{align}
and therefore, we obtain
\begin{align}
    \sum_{i\in[m]} \TD \left( \rho_{\regR_1...\regR_i}, \cR_{\mu,\Tr_{\regR_i}}(\rho_{\regR_1...\regR_{i-1}}) \right)^2 \le O(s) 
\end{align}
by using the Fuchs-van de Graaf inequality together with the domination property of $\mu^{p(s,t)}_{\regR_1...\regR_m}$.
Thus, for a uniformly random round $i\gets[m]$, 
\begin{align}
    \mathbb{E}_i\left[ \TD \left( \rho_{\regR_1...\regR_i}, \cR_{\mu,\Tr_{\regR_i}}(\rho_{\regR_1...\regR_{i-1}}) \right)^2 \right] \le \frac{O(s)}{m}. 
\end{align}
Requiring the learner to recover the target state within trace distance $\epsilon$ with probability at least $1-\delta$ gives $m\ge Cs\epsilon^{-2}\delta^{-1}$ for some constant $C$ and this establishes the upper bound on the round complexity in \cref{thm:main}.

For a general sequence of instruments $(\cI_1,...,\cI_m)$, we incorporate the effects of the measurements into the above strategy and construct an algorithm that achieves the same upper bound for any sequence of quantum instruments.

\subsection{Non-i.i.d. State Tomography}
Next, we explain the information-theoretic upper bound for non-i.i.d. state tomography in \cref{thm:main_tom}.
Recall that state tomography requires the learner to output a classical description of the target unknown quantum state.
The inductive inference algorithm described above outputs only a single copy of a quantum state, so it does not directly provide the classical description required for tomography.
Interestingly, however, the rotated Petz recovery map still plays a central role in our tomography result.

Our proof has two steps.
First, for each distribution $p$ over source programs, we construct a learner $\cA_p$ whose error is small on average over $T\gets p$. 
Second, we combine these learners into a single learner whose error is small for every fixed source program $T$.

Let $\cT$ be the set of Turing machines described by at most $s$ bits, and let $\sigma^T$ denote the state specified by $T\in\cT$.
Fix a distribution $p$ on $\cT$, and define
\begin{align}
    \Omega_{\regA\regR_1...\regR_m}
    \coloneqq \sum_{T\in\cT} p(T)|T\rangle\langle T|_{\regA} \otimes \sigma^T_{\regR_1...\regR_m}.
\end{align}
The classical register $\regA$ stores the source program.
Write $\regR_{<i}\coloneqq\regR_1...\regR_{i-1}$ and $\regR_{\le i}\coloneqq \regR_1...\regR_i$.

For a fixed prior distribution $p$, our learner $\cA_p$ applies a POVM $M$ to the given prefix and, upon obtaining $x$, outputs a classical description of $\Omega_{\regR_i|M,x}$.
Here, $\Omega_{\regR_i\mid M,x}$ is the state of register $\regR_i$ conditioned on obtaining outcome $x$ when $M$ is applied to the prefix $\regR_{<i}$ of the mixture $\Omega$.
Once $M$ is fixed, this description can be computed from $p$ and $x$. 
The remaining task is to choose $M$ so that this mixture is close to $\sigma^T_{\regR_i|M,x}$ on average over $T\gets p$ and the measurement outcomes.

We choose $M$ to approximately minimize the average entropy of the next conditional state. 
To control entropy changes under approximate recovery, we evaluate entropy after applying the depolarizing channel $\cD(\omega)\coloneqq\omega/2+I/2^{n+1}$. 
To analyze the state of the next register conditioned on the outcome of $M$, we use Petz recovery to construct another measurement on the prefix. The near-minimality of $M$ lets us compare the average entropies of the depolarized conditional states produced by these two measurements. We will use this comparison, together with the recovery guarantee, to bound the learner’s average prediction error.

More precisely, let $S$ denote the von Neumann entropy and for a tolerance $\xi>0$, choose a POVM $M$ on $\regR_{<i}$ satisfying
\begin{align}
    \mathbb{E}_x S\left(\cD(\Omega_{\regR_i|M,x})\right)
    \le
    \inf_{M'}\mathbb{E}_y S\left(\cD(\Omega_{\regR_i|M',y})\right) +\xi.
    \label{eq:overview_tom_1}
\end{align}
The infimum ranges over POVMs on $\regR_{<i}$, and the expectations use the outcome distributions obtained by measuring $\Omega_{\regR_{<i}}$. 
To analyze $\cA_p$, let $\cR$ be the averaged rotated Petz map for the reference state $\Omega_{\regA\regR_{<i}}$ and the channel tracing out $\regA$. 
It satisfies $\cR(\Omega_{\regR_{<i}})=\Omega_{\regA\regR_{<i}}$ and
\begin{align}
    \TD\left(
        \Omega_{\regA\regR_{\le i}},
        (\cR\otimes\mathrm{id}_{\regR_i})
        (\Omega_{\regR_{\le i}})
    \right)
    \le \sqrt{I(\regA:\regR_i\mid \regR_{<i})_\Omega}.
    \label{eq:overview_recovery}
\end{align}
Consider the measurement $M'$ that applies $\cR$, measures $\regA$ to obtain $T$, and then applies $M$ to $\regR_{<i}$ to obtain $x$. 
Its output is $(T,x)$. 
The exact recovery identity implies that, on $\Omega$, this output has the same distribution as the experiment
\begin{align}
    T\gets p \text{ and } x\gets M(\sigma^T_{\regR_{<i}}).
\end{align}
By \cref{eq:overview_recovery}, the corresponding conditional next state $\Omega_{\regR_i|M',(T,x)}$ and $\sigma^T_{\regR_i|M,x}$ are close on average. 
To bound the corresponding entropy difference, we use the fact that every eigenvalue of a depolarized
state is at least $2^{-(n+1)}$. 
This bounds the logarithm appearing in the entropy and yields the linear continuity bound
\begin{align}
    |S(\cD(\omega))-S(\cD(\nu))| \le \frac{n+1}{2} \TD(\omega,\nu)
    \label{eq:overview_deporalize}
\end{align}
for all $n$-qubit state $\omega$ and $\nu$.
Using $M'$ as a comparison measurement in \cref{eq:overview_tom_1}, together with this bound, gives
\begin{align}
\mathbb{E}_{T,x}
D\!\left(
\mathcal{D}(\sigma^T_{\regR_i\mid M,x})
\,\middle\|\,
\mathcal{D}(\Omega_{\regR_i\mid M,x})
\right)
&=
\mathbb{E}_x
S\!\left(\mathcal{D}(\Omega_{\regR_i\mid M,x})\right)
-
\mathbb{E}_{T,x}
S\!\left(\mathcal{D}(\sigma^T_{\regR_i\mid M,x})\right)
\\
&\le
\mathbb{E}_{T,x}\!\left[
S\!\left(\mathcal{D}(\Omega_{\regR_i\mid M',(T,x)})\right)
-
S\!\left(\mathcal{D}(\sigma^T_{\regR_i\mid M,x})\right)
\right]
+\xi
\quad \text{(by \cref{eq:overview_tom_1})}
\\
&\le
\frac{n+1}{2}\,
\mathbb{E}_{T,x}
\operatorname{TD}\!\left(
\Omega_{\regR_i\mid M',(T,x)},
\sigma^T_{\regR_i\mid M,x}
\right)
+\xi
\quad \text{(by \cref{eq:overview_deporalize})}
\\
&\le
\frac{n+1}{2}
\sqrt{I(\regA:\regR_i\mid \regR_{<i})_\Omega}
+\xi
\quad \text{(by \cref{eq:overview_recovery})}
\\
&=
O\!\left(
n\sqrt{I(\regA:\regR_i\mid \regR_{<i})_\Omega}
\right)
+\xi.
\end{align}

Here and below, $\mathbb{E}_{T,x}$ uses the experiment above. 
Pinsker’s inequality and $\TD(\cD(\omega),\cD(\nu))=\TD(\omega,\nu)/2$ imply
\begin{align}
    \mathbb{E}_{T,x}
    \TD\left(
        \sigma^T_{\regR_i| M,x},
        \Omega_{\regR_i| M,x}
    \right)^2
    \le
    O\left(n\sqrt{I(\regA:\regR_i|\regR_{<i})_\Omega}\right) +O(\xi).
    \label{eq:overview_mutual}
\end{align}
The chain rule for mutual information bounds the total information about the source program by
\begin{align}
    \sum_{i=1}^m I(\regA:\regR_i|\regR_{<i})_\Omega
    =
    I(\regA:\regR_1...\regR_m)_\Omega
    \le H(p)
    \le \log|\cT|
    \le s+1.
    \label{eq:overview_entropy_bound}
\end{align}

For simplicity, assume that the learner knows the total number of rounds $m$. 
Choosing $\xi=1/m$, summing \cref{eq:overview_mutual} over all rounds, and applying Cauchy–Schwarz with \cref{eq:overview_entropy_bound} gives
\begin{align}
    \sum_{i=1}^m \mathbb{E}_{T,x} \TD \left( \sigma^T_{\regR_i|M,x}, \Omega_{\regR_i|M,x} \right)^2
    \le O(n\sqrt{sm}).
\end{align}
The measurement $M$ in each summand is the one chosen by $\cA_p$ for that round. 
Dividing by $m$ gives expected squared error $O(n\sqrt{s/m})$ for $T\gets p$ and a uniformly random round. 
Hence, for sufficiently large $m$, Markov’s inequality gives an $\epsilon$-accurate prediction with probability at least $1-\delta$, where the probability also includes the measurement outcomes. 
This guarantee holds for every prior $p$, but the learner depends on $p$, and its error on an individual $T$ may be large.

We use multiplicative weight updates to obtain a single learner with the same asymptotic bound for every fixed $T$. 
Start with uniform weights on $\cT$. 
Repeatedly construct $\cA_p$ for the current weighted distribution $p$, evaluate its total expected squared error on each candidate source, and increase the weights of sources with larger error. 
The multiplicative weights analysis shows that, after sufficiently many updates, the average total error of the
constructed learners is $O(n\sqrt{s/m})$ for every $T$.

Our final learner chooses one of these learners uniformly at random and runs it. 
Including this choice in the measurement record makes its expected error equal to the average of their errors. 
Thus, for every target source,
\begin{align}
    \sum_{i=1}^m
    \mathbb{E}_{h,\sigma}
    \TD\left(\rho_{\regR_i| h},\sigma\right)^2
    \le O(n\sqrt{sm}),
\end{align}
where $h$ records the measurement procedure and its full outcome, including the choice of learner, and $\sigma$ is the state described by the output. 
For a uniformly random round, Markov’s inequality gives
\begin{align}
    \Pr_{i,h,\sigma}\left[
        \TD\left(\rho_{\regR_i| h},\sigma\right)>\epsilon
    \right]
    \le
    O\left(\frac{n}{\epsilon^2}\sqrt{\frac{s}{m}}\right).
\end{align}
Making this probability at most $\delta$ yields the round complexity $O(sn^2/(\epsilon^4\delta^2))$ claimed in \cref{thm:main_tom}.
The full proof in \cref{sec:tomography} removes the assumption that $m$ is known without changing this bound.

\subsection{Average-Case Computational Hardness}
We finally explain the computational hardness of universal quantum inductive inference.
In this overview, we only sketch the proof of \cref{thm:main_povm}, from which \cref{thm:main_hard} follows directly.

We first define a special case of universal quantum inductive inference, which we call \emph{POVM-based universal quantum inductive inference}.
In this variant, the sequence of measurements $(M_1,...,M_m)$ is formalized as a sequence of POVMs.
At round $i$, the learner's task is to predict the classical distribution of outcomes obtained by applying $M_i$ to the target state $\rho_{\regR_i| x_1,...,x_{i-1}}$ that is conditioned on the previous measurement outcomes $(x_1,\ldots,x_{i-1})$.
The learner is given the past measurement results $(M_1,x_1,...,M_{i-1},x_{i-1})$ together with the target measurement $M_i$, but is not given any post-measurement quantum states.
Moreover, we consider the average-case version of this task.
Namely, there is a quantum polynomial-time (QPT) algorithm $\cG$ that samples a description of a Turing machine $T$ together with a sequence of measurements $M=(M_1,\ldots,M_m)$.
Let $\rho^T_{\regR_1\cdots\regR_m}$ denote the quantum state whose preparation circuit is generated by $T$ within $t$ steps.
For a randomly chosen round $i\in[m]$, after observing the outcomes $(x_1,...,x_{i-1})$ obtained by applying the first $i-1$ measurements $(M_1,...,M_{i-1})$ to $\rho^T_{\regR_1...\regR_m}$, the learner is required to predict, within statistical distance error $\epsilon$, the distribution of the outcome of the $i$th measurement $M_i$ applied to $\rho^T_{\regR_i| x_1,...,x_{i-1}}$, with probability at least $1-\delta$.

We prove that the hardness of this task is equivalent to the existence of one-way puzzles (OWPuzzs)\footnote{
    Strictly speaking, we prove an equivalence between the existence of infinitely-often OWPuzzs and the average-case hardness of POVM-based universal quantum inductive inference.
    An infinitely-often OWPuzz is defined in the same way as a standard OWPuzz, except that security is required to hold for infinitely many security parameters $\secp$, rather than for all sufficiently large $\secp$.
    This distinction arises because we require the learner for the inductive inference task to succeed for all sufficiently large parameter regimes.
    Since this point is not important for understanding the high-level proof idea, we use the standard notion of OWPuzz throughout this overview.
}.
An OWPuzz consists of a pair of algorithms $(\Samp,\Ver)$.
The sampler $\Samp$ is a QPT algorithm that outputs a pair of classical strings $(\puzz,\ans)$, while $\Ver$ is a computationally unbounded algorithm that checks the correctness of a candidate pair $(\puzz,\ans')$.
Security requires that no QPT adversary, given only $\puzz$, can output an $\ans'$ such that $(\puzz,\ans')$ is accepted by $\Ver$, except with negligible probability.
We establish the equivalence in two steps.
First, assuming the existence of an OWPuzz $(\Samp,\Ver)$, we construct a QPT sampler $\cG$ that generates instances of POVM-based universal quantum inductive inference that are hard on average.
Second, assuming that OWPuzzs do not exist, we construct a QPT learner that solves average-case POVM-based inductive inference.
We explain these two directions below.

\paragraph{The existence of OWPuzzs implies the hardness of inductive inference.}
To prove this direction, we construct a sampler $\cG$ that outputs a Turing machine $T$ together with a sequence of computational basis measurements $(M_1,...,M_m)$.
The machine $T$ generates a circuit description of a quantum state $\rho_{\regR_1...\regR_m}$ such that, for each $i\in[m/2]$, measuring the pair of registers $(\regR_{2i-1},\regR_{2i})$ in the computational basis yields an independent sample $(\puzz,\ans)\gets\Samp(1^\secp)$.

Thus, at every even round $2i$, the learner is given the puzzle $\puzz$ obtained in the preceding round and is required to predict the conditional distribution of $\ans$ given $\puzz$.
If an efficient learner could approximate this conditional distribution, then by running such learner, we could efficiently produce a valid answer $\ans$ for a given puzzle $\puzz$ with non-negligible probability.
This contradicts the security of the OWPuzz $(\Samp,\Ver)$.
Therefore, POVM-based universal quantum inductive inference is hard on average.

\paragraph{The non-existence of OWPuzzs implies the easiness of inductive inference.}
Next, we explain the basic idea behind the construction of an efficient learner for average-case POVM-based universal quantum inductive inference under the assumption that OWPuzzs do not exist.
By the equivalence between OWPuzzs and their distributional variants \cite{C:ChuGolGra24}, we can assume that distributional OWPuzzs do not exist.
Consequently, for every QPT sampler $(\puzz,\ans)\gets\Samp(1^\secp)$, there exists a QPT algorithm $\mathsf{Inv}$ such that
\begin{align}
    \SD\left( {(\puzz,\ans)}, {(\puzz,\mathsf{Inv}(\puzz))} \right)
    \le \frac{1}{\poly(\secp)}.
\end{align}
To obtain a learner independent of the instance generator, we apply this inversion to a QPT sampler $\cQ$ that is a mixture of instance generators and behaves as follows:
It chooses a generator $\cG$ with probability $2^{-|\cG|}$, runs $(T,M)\gets\cG(1^\secp)$, prepares a state $\rho^T$ whose description is generated by a Turing machine $T$, and applies a measurement $M=(M_1,...,M_m)$ to $\rho^T$ obtaining $(x_1,...,x_m)$.
$\cQ$ outputs $(M_1,x_1,...,M_m,x_m)$.
Then, there exists a QPT algorithm that takes $(M_1,x_1,...,M_{i})$ as input and samples from the conditional distribution over $x_i$.
Thus, it suffices to show that the output distribution of $\cQ$ is close to the target distribution.
In the sampling procedure $\cQ$, it chooses the true generator $\cG$ with probability $2^{-|\cG|}$.
Conditioned on this event and the event that $\cG$ generates the target Turing machine $T$, the output distribution of $\cQ$ is identical to the target distribution.
By using this fact, we can show that the KL-divergence between $\cQ$ and the target distribution is upper bounded by $O(s(\secp)+g(\secp))$ for any superconstant function $g(\secp)=\omega(1)$.
Then, by Pinsker's inequality and Markov's inequality, we obtain the upper bound on the round complexity $O((s(\secp)+g(\secp))\epsilon^{-2}\delta^{-1})$ if we require that the learner's output distribution is $\epsilon$-close to the target distribution with probability at least $1-\delta$.

\section{Preliminaries}\label{sec:preliminaries}
\subsection{Basic Notations} 
We use standard notations of quantum information, quantum computation and cryptography.
$\log x$ means $\log_2 x$ and $\ln x$ means $\log_e x$.
For a bit string $x$, $|x|$ is its length.
For bit strings $x$ and $y$, $x\|y$ is their concatenation.
$[n]$ means the set $\{1,2,...,n\}$.
For a finite set $S$, $x\gets S$ means that an element $x$ is sampled uniformly at random from $S$.
For an algorithm $\cA$, $y\gets \cA(x)$ means that the algorithm $\cA$ outputs $y$ on input $x$.
$\negl$ is a negligible function, and $\poly$ is a polynomial.
PPT stands for (classical) probabilistic polynomial-time and QPT stands for quantum polynomial-time. 
For two quantum states $\rho$ and $\sigma$, $\TD(\rho,\sigma)\coloneqq\frac{1}{2}\|\rho-\sigma\|_1$ means their trace distance and $F(\rho,\sigma)\coloneqq\|\sqrt{\rho}\sqrt{\sigma}\|_1$
is the fidelity between $\rho$ and $\sigma$, where $\|X\|_1\coloneqq\Tr\sqrt{X^\dagger X}$ is the trace norm.
The von Neumann entropy of a quantum state $\rho$ is $S(\rho)\coloneqq -\Tr[\rho\log\rho]$.
The relative entropy between two states $\rho$ and $\sigma$ is defined by $D(\rho\|\sigma)\coloneq \Tr [\rho(\log\rho-\log\sigma)]$, where $D(\rho\|\sigma)=\infty$ if $\mathrm{supp}(\rho)\nsubseteq\mathrm{supp}(\sigma)$.
For a bipartite state $\rho_{\regA\regB}$, its quantum mutual information is $I(\regA:\regB)_\rho\coloneqq S(\rho_{\regA})+S(\rho_{\regB})-S(\rho_{\regA\regB})$.
For a tripartite state $\rho_{\regA\regB\regC}$, its conditional quantum mutual information is $I(\regA:\regB|\regC)_\rho\coloneqq S(\rho_{\regA\regC})+S(\rho_{\regB\regC})-S(\rho_{\regC})-S(\rho_{\regA\regB\regC})$. 
For two probability distributions $P$ and $Q$, $\SD(P,Q)\coloneqq\frac{1}{2}\sum_x|P(x)-Q(x)|$ is their statistical distance.
The Shannon entropy of a distribution $P$ is $H(P)\coloneqq -\sum_xP(x)\log P(x)$ and the KL-divergence between distributions $P$ and $Q$ is $\mathrm{KL}(P\|Q)\coloneqq \sum_x P(x)\log \frac{P(x)}{Q(x)}$, where $\mathrm{KL}(P\|Q)\coloneqq \infty$ if $\mathrm{supp}(P)\not\subseteq\mathrm{supp}(Q)$.
Throughout this paper, we fix an arbitrary finite universal quantum gate set $G$ and a binary encoding of quantum circuits over $G$.
For a Turing machine $T$ and a time bound $t\in\N$, $T^t$ denotes the $t$-time execution of $T$.
For Hermitian matrices $A$ and $B$, $A\ge B$ means that $A-B$ is positive semidefinite.

\subsection{Quantum Information}

\begin{definition}[Quantum Instrument]
    Let $X$ be a finite set.
    A quantum instrument $\cI=\{\cE_x\}_{x\in X}$ is a collection of completely-positive maps such that $\sum_{x\in X} \cE_x$ is trace-preserving.
    A measurement channel $\cM_\cI$ with respect to the instrument $\cI$ is defined by
    \begin{align}
        \cM_\cI(\rho) \coloneq \sum_{x\in X} |x\rangle\langle x| \otimes \cE_x(\rho),
    \end{align}
    and a post-measurement state $\rho_{|\cI,x*}$ conditioned on the measurement outcome $x^*\in X$ is 
    \begin{align}
        \rho_{|\cI,x^*} \coloneq \frac{\cE_{x^*}(\rho)}{\Tr[\cE_{x^*}(\rho)]}
    \end{align}
\end{definition}

\begin{definition}[(Rotated) Petz Recovery Map \cite{Petz86,Petz88,Wil15}]
    Let $\cN$ be a quantum channel from a register $\regA$ to a register $\regB$.
    The Petz recovery map $\cP_{\sigma_{\regA},\cN}$ with respect to a reference state $\sigma_{\regA}$ and a quantum channel $\cN$ is defined by
    \begin{align}
        \cP_{\sigma_{\regA},\cN}(X_{\regB}) \coloneqq \sigma_{\regA}^{\frac{1}{2}} \cN^\dagger\left( \cN(\sigma_{\regA})^{-\frac{1}{2}} X_{\regB} \cN(\sigma_{\regA})^{-\frac{1}{2}} \right) \sigma_{\regA}^{\frac{1}{2}}.
    \end{align}
    Moreover, for all $t\in\mathbb{R}$, the rotated Petz recovery map $\cR^t_{\sigma_{\regA},\cN}$ is defined by
    \begin{align}
        \cR^t_{\sigma_{\regA},\cN}(X_{\regB}) \coloneqq \sigma_{\regA}^{-it} \cP_{\sigma_{\regA},\cN}\left( \cN(\sigma_{\regA})^{it} X_{\regB} \cN(\sigma_{\regA})^{-it} \right) \sigma_{\regA}^{it}.
    \end{align}
\end{definition}

\begin{lemma}[\cite{JRSWW18}]
\label{lem:Petz_recover}
    For any quantum channel $\cN$ from $\regA$ to $\regB$ and any quantum states $\rho_{\regA}$ and $\sigma_{\regA}$ satisfying $\mathrm{supp}(\rho_{\regA})\subseteq\mathrm{supp}(\sigma_{\regA})$, 
    \begin{align}
        D(\rho_{\regA}\|\sigma_{\regA}) - D(\cN(\rho_{\regA})\|\cN(\sigma_{\regA})) \ge -2 \log F \left(\rho_{\regA},\cR_{\sigma_{\regA},\cN}\circ\cN(\rho_{\regA}) \right),
    \end{align}
    where
    \begin{align}
        \cR_{\sigma_{\regA},\cN}(\cdot) \coloneqq \frac{\pi}{2}\int_{\mathbb{R}} dt \left( \cosh(\pi t)+1 \right)^{-1} \cR^{\frac{t}{2}}_{\sigma_{\regA},\cN}(\cdot).
    \end{align}
\end{lemma}

We will use the following lemma in the proof of non-i.i.d. state tomography.
It bounds the entropy difference between depolarized states in terms of the trace distance between the original states. 
\begin{lemma}
\label{lem:depolarize}
    Let $\cD$ be an $n$-qubit depolarizing channel defined by
    \begin{align}
        \cD(\omega) \coloneqq \frac{\omega}{2} + \frac{I}{2^{n+1}}
    \end{align}
    for any $n$-qubit state $\omega$.
    Then, for any $n$-qubit states $\omega$ and $\nu$,
    \begin{align}
        |S(\cD(\omega))-S(\cD(\nu))| 
        \le \frac{n+1}{2} \TD(\omega,\nu).
    \end{align}
\end{lemma}
\begin{proof}
For a Hermitian matrix $A$, let $\lambda_{\max}(A)$ and $\lambda_{\min}(A)$ denote its maximal and minimal eigenvalues.
Then by the definition of $\cD$, 
\begin{align}
    \lambda_{\max}(\cD(\nu)) \le \frac{1}{2}+\frac{1}{2^{n+1}}, 
    \text{ and }
    \lambda_{\min}(\cD(\nu)) \ge \frac{1}{2^{n+1}}.
\end{align}
Let $c\coloneqq(\log\lambda_{\max}(\cD(\nu))+\log\lambda_{\min}(\cD(\nu)))/2$.
Then,
\begin{align}
 S(\cD(\omega))-S(\cD(\nu))
    &= -D(\cD(\omega)\|\cD(\nu)) - \Tr [(\cD(\omega)-\cD(\nu)) \log\cD(\nu)] \\
    &\le -\Tr[(\cD(\omega)-\cD(\nu)) \log\cD(\nu)] \\
    &= -\Tr[(\cD(\omega)-\cD(\nu)) (\log\cD(\nu) - cI)] \\
    &\le \|\cD(\omega)-\cD(\nu)\|_1 \|\log\cD(\nu) - cI\|_\infty \label{eq:tom_holder} \\
    &= 2\TD(\cD(\omega),\cD(\nu)) \|\log\cD(\nu) - cI\|_\infty \\
    &= \TD(\cD(\omega),\cD(\nu)) (\log\lambda_{\max}(\cD(\nu))-\log\lambda_{\min}(\cD(\nu))) \\
    &\le \log (2^n+1) \TD(\cD(\omega),\cD(\nu)) \\
    &= \frac{\log (2^n+1)}{2} \TD(\omega,\nu).
\end{align}
In \cref{eq:tom_holder}, we used H\"older's inequality and $\|A\|_\infty\coloneqq \sqrt{\lambda_{\max}(A^\dagger A)}$ denotes the operator norm of $A$.
$\omega$ and $\nu$ are symmetric in the above argument and therefore we have
\begin{align}
    |S(\cD(\omega))-S(\cD(\nu))| 
        \le \frac{\log (2^n+1)}{2} \TD(\omega,\nu).
\end{align}
Finally, we use $\log (2^n+1)\le n+1$ and complete the proof.
\end{proof}

\subsection{Quantum Algorithmic Information}
We introduce the notion of \emph{$t$-time-bounded universal density matrix}.
This is a time-bounded variant of universal (semi-)density matrix introduced by G\'acs~\cite{Gac01}.

\begin{definition}[$t$-Time-Bounded Universal Density Matrix]
    Let $U$ be a universal (classical) Turing machine and let $t\in\N$ be a time bound.
    The $n$-qubit $t$-time-bounded universal density matrix with respect to $U$ is
    \begin{align}
        \mu^{U,t}_n \coloneq 2^{-t}\sum_{\Pi\in\bit^t}\rho^{U,t}_n(\Pi),
    \end{align}
    where $\rho^{U,t}_n(\Pi)$ is an $n$-qubit state defined as follows:
    Let $\sigma(\Pi)$ be an $\ell_\Pi$-qubit state whose preparation circuit is generated by $U^t(\Pi)$. 
    Then,
    \begin{align}
        \rho^{U,t}_n(\Pi) \coloneqq 
        \begin{cases}
            \Tr_{\ell_\Pi-n} \sigma(\Pi) & \text{if $\ell_\Pi\ge n$} \\
            \sigma(\Pi)\otimes |0^{n-\ell_\Pi}\rangle\langle 0^{n-\ell_\Pi}| & \text{if $\ell_\Pi<n$},
        \end{cases}
        \label{eq:def_padding}
    \end{align}
    where $\Tr_{\ell_\Pi-n}$ denotes the partial trace over the last $\ell_\Pi-n$ qubits.
    Moreover, if $U^t(\Pi)$ does not output a valid circuit description, we set $\rho^{U,t}_n(\Pi)=|0^n\rangle\langle 0^n|$.
\end{definition}
Note that the partial trace of an $n$-qubit time-bounded universal density matrix over its last $k$ qubits is identical to the $(n-k)$-qubit time-bounded universal density matrix, namely $\mu^{U,t}_{n-k}=\Tr_k\mu^{U,t}_n$.
In the following, whenever we say that a Turing machine generates a circuit description of an $n$-qubit state, if the circuit does not output exactly $n$ qubits, we implicitly apply partial tracing or padding, as in \cref{eq:def_padding}, to adjust the number of output qubits to $n$.
Then, we prove the following domination property.
\begin{lemma}\label{lem:domination}
    There exists a universal Turing machine $U$, a polynomial $p:\N^2\to\N$, and a constant $c>0$ such that for all $n,t\in\mathbb N$, all Turing machines $T$, and all $\tau\ge p(|T|,t)$,
    \begin{align}
        \mu_n^{U,\tau} \ge 2^{-c|T|} \rho_n^{T,t},
    \end{align}
    where $|T|$ denotes the length of the binary description of $T$ and $\rho_n^{T,t}$ is an $n$-qubit quantum state whose preparation circuit is generated by $T^t$.
\end{lemma}

\begin{proof}[Proof of \cref{lem:domination}]
For a Turing machine $T$, let $\Pi_T=1^{|k|}\|0\|k\|T$ be its self-delimiting encoding, where $k$ is a binary representation of $|T|$.
Note that $|\Pi_T|\le |T|+2\log|T|+1\le c|T|$ for some constant $c$.
Then, there exists a time-efficient Turing machine $U$ and a polynomial $p:\N^2\to\N$ such that for any $t\in\N$, any Turing machine $T$, any $\tau\ge p(|T|,t)$, and any $x\in\bit^*$, $U^\tau(\Pi_T\|x)$ simulates $T^t$.
Without loss of generality, we can choose $p$ to satisfy $p(|T|,t)\ge |\Pi_T|$ for all $t$.
By the definition of the time-bounded universal density matrix,
\begin{align}
    \mu_n^{U,\tau}
    &=2^{-\tau} \sum_{\Pi\in\bit^{\tau}} \rho_n^{U,\tau}(\Pi) \\
    &\ge 2^{-\tau} \sum_{r\in\bit^{\tau-|\Pi_{T}|}} \rho_n^{U,\tau}(\Pi_T\|r) \\
    &= 2^{-\tau} \sum_{r\in\bit^{\tau-|\Pi_{T}|}} \rho_n^{T,t} \\
    &= 2^{-|\Pi_{T}|} \rho_n^{T,t}.
\end{align}
Using the fact that $|\Pi_{T}|\le c|T|$ for some constant $c$, we finally obtain
\begin{align}
    \mu_n^{U,\tau} \ge 2^{-c|T|} \rho_n^{T,t}.
\end{align}
\end{proof}

\subsection{Cryptography}

\begin{definition}[Infinitely-Often One-Way Puzzles \cite{STOC:KhuTom24}]
    An infinitely-often one-way puzzle (OWPuzz) is a pair $(\Samp, \Ver)$ of algorithms with the following syntax:
    \begin{itemize}
    \item 
    $\Samp(1^\secp)\to (\puzz,\ans):$
    A QPT algorithm that, on input the security parameter $\secp$, outputs a pair $(\puzz,\ans)$ of classical strings.
    \item
    $\Ver(\puzz,\ans')\to\top/\bot:$
    An unbounded algorithm that, on input $(\puzz,\ans')$, outputs $\top$ or $\bot$.
    \end{itemize}
    We require the following properties.
    \begin{itemize}
    \item \textbf{Correctness:} 
        \begin{align}
            \Pr[\top\gets\Ver(\puzz,\ans):(\puzz,\ans)\gets\Samp(1^\secp)] \ge 1-\negl(\secp).
        \end{align}
    \item \textbf{Security:} For any QPT adversary $\cA$ and any polynomial $p$,
        \begin{align}
            \Pr[\top\gets\Ver(\puzz,\cA(1^\secp,\puzz)):(\puzz,\ans)\gets \Samp(1^\secp)] \le \frac{1}{p(\secp)}
        \end{align}
        holds for infinitely many $\secp\in\N$.
    \end{itemize}
\end{definition}

We also review the definition of infinitely-often distributional one-way puzzles.

\begin{definition}[Infinitely-Often Distributional One-Way Puzzles \cite{C:ChuGolGra24}]
    A QPT algorithm $\Samp$ that takes the security parameter $1^\secp$ as input and outputs a pair $(\puzz,\ans)$ of classical bit strings
    is an infinitely-often distributional one-way puzzle
    if there exists a polynomial $p$ such that for any QPT adversary $\cA$, and for infinitely many $\secp\in\N$,
    \begin{align}
        \SD\left( \{\puzz,\ans\}_{(\puzz,\ans)\gets\Samp(1^\secp)} , \{\puzz,\cA(1^\secp,\puzz)\}_{(\puzz,\ans)\gets\Samp(1^\secp)} \right) \ge \frac{1}{p(\secp)}.
    \end{align}
\end{definition}

Clearly, if $(\Samp,\Ver)$ is an infinitely-often OWPuzz, then $\Samp$ is an infinitely-often distributional OWPuzz.
Chung, Goldin, and Gray \cite{C:ChuGolGra24} showed that (infinitely-often) distributional OWPuzzs imply (infinitely-often) OWPuzzs.
Combining them, the following equivalence is known.

\begin{lemma}[\cite{C:ChuGolGra24}]
    \label{lem:distowpuzz_owpuzz}
    Infinitely-often OWPuzzs exist if and only if infinitely-often distributional OWPuzzs exist.
\end{lemma}

\section{Universal Inductive Inference of Quantum States}
\label{sec:uqii}
In this section, we introduce \emph{universal quantum inductive inference} and establish both information-theoretic and computational results. 
We first define universal quantum inductive inference in \cref{sec:def_uqii}. 
Then, in \cref{sec:uqii_info}, we present an information-theoretic algorithm for solving universal quantum inductive inference. 
Finally, we prove its average-case computational hardness under a quantum cryptographic assumption (\cref{sec:uqii_comp}).

\subsection{Definition of Universal Quantum Inductive Inference}
\label{sec:def_uqii}
We formalize universal quantum inductive inference, in which a learner
observes the preceding registers of a quantum state and predicts the next
register conditioned on the observed data.
Let $n,m,s,t\in\N$ be parameters, and for each $i\in[m]$, let $\regR_i$ be an
$n$-qubit register.
We define the set
\begin{align}
    \cC_{n,m}[s,t]
    :=\left\{
        \rho_{\regR_1...\regR_m}:
        \begin{gathered}
            \text{There exists a classical Turing machine $T$ with $|T|\le s$} \\
            \text{ such that $T^t$ outputs a circuit preparing $\rho_{\regR_1...\regR_m}$}
        \end{gathered}
    \right\}
\end{align}
of $mn$-qubit states whose preparation circuit descriptions are generated
by Turing machines described by at most $s$ bits within $t$ steps.
To specify the observed data, fix a state
$\rho_{\regR_1...\regR_m}\in\mathcal C_{n,m}[s,t]$ and consider a sequence of quantum instruments $(\mathcal I_1,\ldots,\mathcal I_m)$, where each $\mathcal I_j:=\{\mathcal E_x^{(j)}\}_{x\in\mathcal X_j}$ acts on $\regR_j$.
Each instrument is specified by a circuit description implementing it.
We denote its classical outcome register by $\regX_j$ and its $n$-qubit
post-measurement register by $\regR_j$.
Let $\rho_{<i}$ and $\rho_{\le i}$ be the reduced states on registers $(\regR_1,...\regR_{i-1})$ and $(\regR_1,...,\regR_i)$.
For a fixed round $i$, let $(x_1,\ldots,x_{i-1})$ be the classical
outcomes obtained by applying $(\mathcal I_1,\ldots,\mathcal I_{i-1})$
to the first $i-1$ registers $(\regR_1,...,\regR_{i-1})$.
Let $h_{<i}:=(\mathcal I_1,x_1,\ldots,\mathcal I_{i-1},x_{i-1})$ and define
\begin{align}
    \mathcal A_{h_{<i}}
    :=\mathcal E_{x_1}^{(1)}\otimes\cdots\otimes
      \mathcal E_{x_{i-1}}^{(i-1)}.
\end{align}
The probability of obtaining $(x_1,...,x_{i-1})$ is
\begin{align}
    \cD^{\cI,\rho}_{<i}(x_1,...,x_{i-1}) \coloneqq \Tr\left[\cA_{h_{<i}}(\rho_{<i})\right].
\end{align}
The learner takes as input $h_{<i}$, the description of the next
instrument $\mathcal I_i$, and a single copy of the post-measurement state
\begin{align}
    \rho_{<i\mid h_{<i}}
    :=\frac{\mathcal A_{h_{<i}}(\rho_{<i})}{\cD^{\cI,\rho}_{<i}(x_1,...,x_{i-1})},
\end{align}
together with the parameters specified below.
Define the conditional state on $(\regR_1,...,\regR_i)$ before the $i$th measurement by
\begin{align}
    \rho_{\le i\mid h_{<i}}
    :=\frac{
        (\cA_{h_{<i}}\otimes\mathrm{id}_{\regR_i})(\rho_{\le i})
    }{\cD^{\cI,\rho}_{<i}(x_1,...,x_{i-1})}.
\end{align}
The prediction target $\tau_{\le i\mid h_{<i}}$ is the state on registers $(\regR_1,...,\regR_i,\regX_i)$ obtained by applying $\mathcal I_i$ to $\regR_i$:
\begin{align}
    \tau_{\le i\mid h_{<i}}
    :=(\mathrm{id}_{\regR_{<i}}\otimes\mathcal M_{\mathcal I_i})
      (\rho_{\le i\mid h_{<i}}),
\end{align}
where $\mathcal M_{\mathcal I_i}$ is the measurement channel associated
with $\mathcal I_i$.
Its marginal on $(\regR_1,...,\regR_{i-1})$ is exactly the learner's input state
$\rho_{<i\mid h_{<i}}$.

Our formal definition is as follows.
\begin{definition}[Universal Quantum Inductive Inference]
\label{def:uqii}
Let $r:\mathbb N^3\to\mathbb N$ be a function.
We say that a quantum algorithm $\mathcal L$ solves universal quantum
inductive inference with round complexity $r$ if, for all
$n,s,t,\epsilon^{-1},\delta^{-1}\in\mathbb N$,
all $m\ge r(s,\epsilon^{-1},\delta^{-1})$,
all $\rho_{\regR_1...\regR_m}\in\mathcal C_{n,m}[s,t]$, and all sequences of instruments $(\mathcal I_1,\ldots,\mathcal I_m)$,
the following holds:
\begin{align}
    \Pr_{\substack{i\leftarrow[m]\\x_{<i}\leftarrow \cD^{\cI,\rho}_{<i}}}
    \left[
        \TD\left(
            \cL(\mathsf{param},h_{<i},\mathcal I_i,
                       \rho_{<i\mid h_{<i}}),
            \tau_{\le i\mid h_{<i}}
        \right)\le\epsilon
    \right]
    \ge 1-\delta,
\end{align}
where
$\mathsf{param}:=(1^s,1^t,1^n,1^{\epsilon^{-1}},1^{\delta^{-1}})$
is the collection of known parameters.
The probability is over the uniform choice of $i\in[m]$ and the measurement outcomes
obtained by applying the first $i-1$ instruments to $\rho_{\regR_1...\regR_m}$.
\end{definition}

\subsection{An Information-Theoretic Algorithm for Universal Quantum Inductive Inference}
\label{sec:uqii_info}

\begin{theorem}\label{thm:info_alg}
There exists an algorithm that solves universal quantum inductive inference with round complexity $r(s,\epsilon^{-1},\delta^{-1}) = O(s\epsilon^{-2}\delta^{-1})$.
\end{theorem}

\begin{proof}[Proof of \cref{thm:info_alg}]
We first outline the proof strategy.
The learner receives the quantum state left by the past measurements.
It uses this state to predict the next measurement outcome and
post-measurement state, including their correlations with the past.
We construct this prediction using a recovery map determined by
a universal reference state and the observed measurement history.
For the analysis, we compare the true and reference states on the
registers observed so far, keeping the classical measurement records
as well.
The recovery inequality relates the prediction error, averaged over
histories, to the increase in relative entropy as one more register
is observed.
Using the domination property, we bound the total increase by $O(s)$.
Markov's inequality then bounds the probability of an error exceeding
$\epsilon$ at a uniformly chosen round by $O(s/(m\epsilon^2))$,
which gives the claimed round complexity.
We now give the construction and its error analysis in detail.

Fix $n,s,t,m\in\N$, a sequence of instruments $(\cI_1,\ldots,\cI_m)$, and a state $\rho_{\regR_1,...,\regR_m}\in\cC_{n,m}[s,t]$.
For each $j\in[m]$, let
\begin{align}
    \cI_j = \left\{ \cE^{(j)}_x\right\}_{x\in X_j}.
\end{align}
By the definition of $\cC_{n,m}[s,t]$, there exists a Turing machine $T$ with $|T|\le s$ such that $T^t$ outputs the circuit description preparing $\rho_{\regR_1...\regR_m}$.
Then, \cref{lem:domination} implies that there exists a universal Turing machine $U$, a polynomial $p:\N^2\to\N$, and a constant $c$ such that
\begin{align}
    \mu_{\regR_1...\regR_m}^{U,p(s,t)} \ge 2^{-cs} \rho_{\regR_1...\regR_m}.
\end{align}
In the following, we fix such $U$ and $p$ and write $\mu_{\regR_1...\regR_m}=\mu_{\regR_1...\regR_m}^{U,p(s,t)}$ for notational simplicity.

We construct an inference algorithm $\cL$ as follows:
\begin{enumerate}
    \item Take the following as input:
    \begin{itemize}
        \item Parameters $\mathsf{param}=(1^s,1^t,1^n,1^{\epsilon^{-1}},1^{\delta^{-1}})$.
        \item The collection of measurement results $h_{<i}=(\cI_1,x_1,...,\cI_{i-1},x_{i-1})$.
        \item The description of the next measurement $\cI_i$.
        \item The post-measurement state $\rho_{<i|h_{<i}}$ over the registers $(\regR_1,...,\regR_{i-1})$ conditioned on $h_{<i}$.
    \end{itemize}
    \item Define
    \begin{align}
        \cA_{h_{<i}} \coloneq \cE^{(1)}_{x_1} \otimes...\otimes \cE^{(i-1)}_{x_{i-1}},
        \text{ and }
        \mu_{\le i| h_{<i}} \coloneq \frac{ \left( \cA_{h_{<i}}\otimes \cM_{\cI_i}\right)(\mu_{\le i}) }{ \Tr\left[ \left( \cA_{h_{<i}}\otimes \cM_{\cI_i}\right)(\mu_{\le i}) \right] }.
    \end{align}
    Note that $\mu_{\le i|h_{<i}}$ is a state over registers $(\regR_1,...,\regR_{i},\regX_i)$.
    \item Let $\cN_i \coloneqq \Tr_{\regR_i\regX_i}$ be a partial trace discarding the $i$th register $\regX_i\regR_i$ and set
    \begin{align}
        \mu_{<i\mid h_{<i}} \coloneq \cN_i\left(\mu_{\le i\mid h_{<i}}\right).
    \end{align}
    Let $\cR_{\mu_{\le i| h_{<i}},\cN_i}$ be the recovery map that satisfies \cref{lem:Petz_recover}.
    \item Return $\cR_{\mu_{\le i| h_{<i}},\cN_i} \left( \rho_{<i| h_{<i}}\right)$.
\end{enumerate}
We next analyze the performance of $\cL$.
For each $j\in[m]$, we define
\begin{align}
    \omega^\rho_{j} \coloneqq (\cM_{\cI_1}\otimes...\otimes\cM_{\cI_{j}}) \rho_{\le j}, 
    \quad
    \omega^\mu_{j} \coloneqq (\cM_{\cI_1}\otimes...\otimes\cM_{\cI_{j}}) \mu_{\le j}
\end{align}
and
\begin{align}
    \Delta_j 
    &\coloneq D(\omega^\rho_j\|\omega^\mu_j) - D(\omega^\rho_{j-1}\|\omega^\mu_{j-1}) \\
    &= \sum_{x_1,...,x_{j-1}} \cD^{\cI,\rho}_{<j}(x_1,...,x_{j-1}) \left[ D\left( \tau_{\le j| h_{<j}} \middle\| \mu_{\le j| h_{<j}} \right) - D\left( \rho_{<j| h_{<j}} \middle\| \mu_{<j| h_{<j}} \right) \right]
    \label{eq:delta_relative}, 
\end{align}
where
\begin{align}
    \cD^{\cI,\rho}_{<i}(x_1,...,x_{i-1})\coloneqq \Tr \left[ \cA_{h_{<i}}(\rho_{<i}) \right]
\end{align}
Since
\begin{align}
    \cN_i\left(
        \tau_{\le i\mid h_{<i}}
    \right)
    =
    \rho_{<i\mid h_{<i}}, \text{ and }
    \cN_i\left(
        \mu_{\le i\mid h_{<i}}
    \right)
    =
    \mu_{<i\mid h_{<i}},
\end{align}
we obtain
\begin{align}
    D\left( \tau_{\le i| h_{<i}} \middle\| \mu_{\le i\mid h_{<i}} \right) - D\left(
    \rho_{<i| h_{<i}} \middle\| \mu_{<i| h_{<i}} \right) 
    &\ge -2\ln F\left( \tau_{\le i\mid h_{<i}}, \cR_{ \mu_{\le i| h_{<i}}, \cN_i} \left(\rho_{<i| h_{<i}} \right)\right) \\
    &= -2\ln F\left( \tau_{\le i| h_{<i}}, \cL\left( \mathsf{param},h_{<i},\cI_i,\rho_{<i\mid h_{<i}}\right)\right).
\end{align}
By the Fuchs--van de Graaf inequality and $1-F(\rho,\sigma)^2\le -2\ln F(\rho,\sigma)$, it follows that
\begin{align}
    \label{eq:relative_TD}
    D\left( \tau_{\le i| h_{<i}} \middle\| \mu_{\le i| h_{<i}} \right) - D\left( \rho_{<i| h_{<i}} \middle\| \mu_{<i| h_{<i}} \right) 
    \ge \TD\left( \cL\left( \mathsf{param},h_{<i},\cI_i,\rho_{<i| h_{<i}}\right),\tau_{\le i| h_{<i}} \right)^2.
\end{align}
Combining \cref{eq:delta_relative} and \cref{eq:relative_TD}, we obtain
\begin{align}
    \Delta_i
    \ge \mathbb E_{x_{<i}} \left[ \TD\left(\cL\left(\mathsf{param},h_{<i},\cI_i,\rho_{<i| h_{<i}}\right), \tau_{\le i| h_{<i}} \right)^2 \right].
\end{align}
Then,
\begin{align}
    \sum_{i\in[m]} \mathbb{E}_{x_{<i}} \left[ \TD\left( \cL\left(\mathsf{param},h_{<i},\cI_i,\rho_{<i| h_{<i}}\right),
    \tau_{\le i| h_{<i}}\right)^2 \right]
    &\le \sum_{i\in[m]} \Delta_i\\
    &= D(\omega^\rho_m\|\omega^\mu_m) \\
    &\le D(\rho_{\le m}\|\mu_{\le m}) \\
    &\le cs.
\label{eq:total_error}
\end{align}
Finally, by Markov's inequality,
\begin{align}
    &\Pr_{i,x_{< i}} \left[ \TD\left( \cL\left( \mathsf{param},h_{<i},\cI_i,\rho_{<i| h_{<i}}\right),\tau_{\le i| h_{<i}}\right) > \epsilon \right] \\
    &\quad\le \frac{1}{m\epsilon^2} \sum_{i\in[m]} \mathbb{E}_{x_{<i}} \left[\TD\left( \cL\left(\mathsf{param},h_{<i},\cI_i,\rho_{<i| h_{<i}}\right),\tau_{\le i| h_{<i}}\right)^2 \right]\\
    &\quad\le \frac{cs}{m\epsilon^2}.
\end{align}
Therefore, $\cL$ solves universal quantum inductive inference with round complexity $O(s\epsilon^{-2}\delta^{-1})$.
\end{proof}

\subsection{Average-Case Computational Hardness of Universal Quantum Inductive Inference}
\label{sec:uqii_comp}

We introduce a notion of average-case universal quantum inductive inference.
\begin{definition}[Average-Case Universal Quantum Inductive Inference]
Let $\secp\in\N$ be a parameter and let $r:\N^4\to\N$ be a polynomial.
We say that a QPT algorithm $\cL$ solves average-case universal quantum inductive inference with round complexity $r$ if, for all polynomials $m,n,s,t,t',p,q:\N\to\N$ with $m(\secp)\ge r(\secp,s(\secp),p(\secp),q(\secp))$, all $t'(\secp)$-time algorithms $\cG$, the following holds:
\begin{itemize}
    \item $\cG$ takes $1^\secp$ as input and outputs a Turing machine $T$ such that $|T|\le s(\secp)$ and $T^{t(\secp)}$ generates a circuit description of $\rho_{\regR_1...\regR_m}^T\in\cC_{n(\secp),m}[s(\secp),t(\secp)]$, and a sequence $\cI=(\cI_1,...,\cI_m)$ of instruments, where each $\cI_i$ acts on $\regR_i$.
    \item For all sufficiently large $\secp$,
    \begin{align}
    \Pr_{i,T,\cI,x_{<i}}\left[
        \TD\left(
            \cL(\mathsf{param},h_{<i},\cI_i,\rho^T_{\regR_1...\regR_{i-1}|h_{<i}}),
            \tau^T_{\regR_1...\regR_i|h_{<i}}
        \right)\le \frac{1}{p(\secp)}
    \right]\ge1-\frac{1}{q(\secp)},
    \end{align}
    where $i\gets[m]$, $(T,\cI)\gets\cG(1^\secp)$, $x_{<i}\gets\cD_{<i}^{\cI,\rho}$, $\mathsf{param}\coloneqq(1^\secp,1^{n(\secp)},1^{s(\secp)},1^{t(\secp)},1^{t'(\secp)},1^{p(\secp)},1^{q(\secp)})$, and $h_{<i}\coloneqq(\cI_1,x_1,...,\cI_{i-1},x_{i-1})$.
\end{itemize}
\end{definition}

\begin{theorem}
\label{thm:owpuzz_to_hard}
    If infinitely-often OWPuzzs exist, then for any polynomial $r$, there exists no QPT algorithm that solves average-case universal quantum inductive inference with round complexity $r$.
\end{theorem}

\begin{proof}[Proof of \cref{thm:owpuzz_to_hard}]
Let $(\mathsf{Samp},\mathsf{Ver})$ be an infinitely-often OWPuzz.  
Without loss of generality, we can assume that $\Samp(1^\secp)$ generates $\puzz\in\bit^{n(\secp)}$ and $\ans\in\bit^{n(\secp)}$ for some polynomial $n:\N\to\N$.
Let 
\begin{align}
    \omega_\secp \coloneqq \sum_{\puzz,\ans} \Pr[(\puzz,\ans)\gets\Samp(1^\secp)] |\puzz\rangle\langle\puzz| \otimes |\ans\rangle\langle\ans|
\end{align}
be a $2n(\secp)$-qubit state.
Because $\Samp$ is a QPT algorithm, for any polynomial $m$, there exists a QPT algorithm $\cG$ and polynomials $t,t':\N\to\N$ such that 
\begin{itemize}
    \item For all $\secp$, $\cG(1^\secp)$ runs in time $t'(\secp)$ and outputs a Turing machine $T_\secp$ with $|T_\secp|\le O(\log\secp)$ and a sequence of instruments $(\cI_1,...,\cI_{m(\secp)})$, where each $\cI_i$ is a computational basis measurement and $T_\secp^{t(\secp)}$ generates the circuit description of $\rho^\secp_{\regR_1...\regR_m(\secp)}$ defined by
    \begin{align}
        \rho^\secp_{\regR_1...\regR_{m(\secp)}} \coloneqq 
    \begin{cases}
        \omega_\secp^{\otimes m(\secp)/2} & \text{ if $m(\secp)$ is even} \\
        \omega_\secp^{\otimes (m(\secp)-1)/2}\otimes \sum_{\puzz}\Pr[\puzz\gets\Samp(1^\secp)]|\puzz\rangle\langle\puzz| & \text{ if $m(\secp)$ is odd}.
    \end{cases}
    \end{align}
\end{itemize}

Our goal is to show the average-case hardness of universal quantum inductive inference.
For the sake of contradiction, assume that there exists a QPT algorithm $\cL$ that solves average-case universal inductive inference with round complexity 
$r(\secp,s(\secp),p(\secp),q(\secp))$ for some polynomial $r:\N^4\to\N$.
Then, for any $m(\secp)\ge r(\secp,\secp,2,4)$ and the QPT algorithm $\cG$ defined above, $\cL$ satisfies 
\begin{align}
    \Pr_{i,x_{<i}} \left[ \TD \left(\cL(\param_\secp,h_{<i},\cI_i,\rho^\secp_{\regR_1...\regR_{i-1}|h_{<i}}), \tau^\secp_{\regR_1...\regR_{i}|h_{<i}}) \right) \le  \frac{1}{2} \right] \ge \frac{3}{4},
\end{align}
for all sufficiently large $\secp$.
By the definition of $\rho^\secp_{\regR_1...\regR_{m(\secp)}}$ and $(\cI_1,...,\cI_{m(\secp)})$, for each $j\in[\lfloor m(\secp)/2 \rfloor ]$, $(x_{2j-1},x_{2j})$ is sampled with probability $\Pr[(x_{2j-1},x_{2j})\gets\Samp(1^\secp)]$.
Moreover, the post-measurement state $\rho^\secp_{\regR_1...\regR_{i-1}|h_{<i}}$ conditioned on $h_{<i}$ is $\rho^\secp_{\regR_1...\regR_{i-1}|h_{<i}}=|x_1,...,x_{i-1}\rangle\langle x_1,...,x_{i-1}|$ and the target state is 
\begin{align}
    \tau^\secp_{\regR_1...\regR_{i}\regX_i|h_{<i}} = 
    \begin{cases}
        |x_{<i}\rangle\langle x_{<i}| \otimes \sum_{x_i} \Pr[(x_{i-1},x_i)\gets\Samp(1^\secp)|x_{i-1}] |x_i\rangle\langle x_i|^{\otimes 2} & \text{ if $i$ is even} \\
        |x_{<i}\rangle\langle x_{<i}| \otimes \sum_{x_i,x_{i+1}} \Pr[(x_{i},x_{i+1})\gets\Samp(1^\secp)] |x_i\rangle\langle x_i|^{\otimes 2} & \text{ if $i$ is odd}.
    \end{cases}
\end{align}
Then, for $a\coloneqq\lfloor m(\secp)/2\rfloor$,
\begin{align}
    \frac{1}{a}\sum_{i\in[a]} \Pr_{x_{<2i}} \left[ \TD \left(\cL(\param_\secp,h_{<2i},\cI_{2i},\rho^\secp_{\regR_1...\regR_{2i-1}|h_{<2i}}), \tau^\secp_{\regR_1...\regR_{2i}|h_{<2i}}) \right) > \frac{1}{2} \right] 
    &\le \frac{m(\secp)}{4a} \\
    &\le \frac{3}{4}.
\end{align}
In the last inequality, we assumed $m(\secp)\ge 2$ without loss of generality and used $a=\lfloor m(\secp)/2\rfloor\ge m(\secp)/3$.
Moreover,
\begin{align}
    &\frac{1}{a} \sum_{i\in[a]} \mathbb{E}_{x_{<2i}} \left[ \TD \left(\cL(\param_\secp,h_{<2i},\cI_{2i},\rho^\secp_{\regR_1...\regR_{2i-1}|h_{<2i}}), \tau^\secp_{\regR_1...\regR_{2i}|h_{<2i}} \right) \right] \\
    &\le \frac{1}{2} + \frac{1}{2a} \sum_{i\in[a]} \Pr_{x_{<2i}} \left[ \TD \left(\cL(\param_\secp,h_{<2i},\cI_{2i},\rho^\secp_{\regR_1...\regR_{2i-1}|h_{<2i}}), \tau^\secp_{\regR_1...\regR_{2i}|h_{<2i}} \right) > \frac{1}{2} \right] \\
    &\le \frac{1}{2} + \frac{3}{8} = \frac{7}{8}
\end{align}
for all sufficiently large $\secp$.
By using such $\cL$, we construct a QPT adversary $\cA$ that breaks the security of infinitely-often OWPuzz $(\Samp,\Ver)$ as follows:
\begin{enumerate}
    \item Take $1^\secp$ and $\puzz$ as input.
    \item Sample $i\gets[a]$, where $a=\lfloor \frac{m(\secp)}{2}\rfloor$.
    \item For each $j\in[i-1]$, run $(\puzz_j,\ans_j)\gets\Samp(1^\secp)$ and let $x_{2j-1}=\puzz_j$ and $x_{2j}=\ans_j$.
    \item Let $x_{2i-1}=\puzz$ and run $\cL(\param_\secp,h_{<2i},\cI_{2i},\rho^\secp_{\regR_1...\regR_{2i-1}|h_{<2i}})$ to obtain a state $\sigma_{\regR_1...\regR_{2i}\regX_{2i}}$.
    \item Measure $\regX_{2i}$ in the computational basis and output the result.
\end{enumerate}
Then,
\begin{align}
    &\Pr\left[\top\gets\Ver(\puzz,\cA(1^\secp,\puzz)): (\puzz,\ans)\gets\Samp(1^\secp) \right] \\
    &\ge \Pr\left[\top\gets\Ver(\puzz,\ans) \right] - \SD(\{\puzz,\cA(1^\secp,\puzz)\},\{\puzz,\ans\}) \\
    &\ge 1-\negl(\secp) - \frac{1}{a} \sum_{i\in[a]} \mathbb{E}_{x_{<2i}} \left[ \TD \left(\cL(\param_\secp,h_{<2i},\cI_{2i},\rho^\secp_{\regR_1...\regR_{2i-1}|h_{<2i}}), \tau^\secp_{\regR_1...\regR_{2i}|h_{<2i}} \right) \right] \\
    &\ge \frac{1}{8} - \negl(\secp)
\end{align}
holds for all sufficiently large $\secp$ and $\cA$ breaks the security of infinitely-often OWPuzzs.
Therefore, we complete the proof.
\end{proof}

\section{POVM-Based Universal Quantum Inductive Inference}
\label{sec:POVM}

In this section, we investigate a variant of universal quantum inductive inference.
Universal quantum inductive inference, as defined in the previous section, requires a learner to make accurate predictions with high probability for any sequence of quantum instruments. 
In particular, for a target quantum instrument, the learner must predict both the classical measurement outcome and the corresponding post-measurement quantum state. 
Therefore, when one is interested only in the classical distribution of measurement outcomes, universal quantum inductive inference imposes a stronger requirement than necessary.
Motivated by this observation, we define a variant of inductive inference based on POVMs, that we call \emph{POVM-based universal quantum inductive inference}. 
In this setting, given a sequence of POVMs, the learner predicts the classical distributions of their outcomes. 
Moreover, the measurement history available to the learner is entirely classical, in particular, the learner is not given access to any post-measurement quantum states.

We formalize our definition and give an information-theoretic upper bound in \cref{sec:def_POVM}.
Moreover, we prove the equivalence between the average-case hardness of POVM-based universal quantum inductive inference and infinitely-often OWPuzzs (\cref{sec:comp_povm}). 

\subsection{Definition of POVM-Based Universal Quantum Inductive Inference}\label{sec:def_POVM}
For each $i\in\N$, let $\regR_i$ be an $n$-qubit register.
Let $\rho_{\regR_1...\regR_m}$ be an $mn$-qubit state and let $M=(M_1,...,M_m)$ be a sequence of POVMs, where each $M_j=\{E_{x}^{(j)}\}_{x}$ acts on $\regR_j$.
Then, we define a distribution $\cD^{M,\rho}_{\le m}$ obtained by applying $M$ to $\rho_{\regR_1...\regR_m}$ as follows:
\begin{align}
    \cD_{\le m}^{M,\rho}(x_1,...,x_m) \coloneqq \Tr\left[ (E^{(1)}_{x_1}\otimes ... \otimes E^{(m)}_{x_m}) \rho_{\regR_1...\regR_m} \right].
\end{align}
Moreover, for each $i\in[m]$, let $\cD_{<i}^{M,\rho}$ be the marginal distribution of $\cD_{\le m}^{M,\rho}$ on the first $i-1$ coordinates, and let $\cD_i^{M,\rho}(\cdot|x_{<i})$ be the conditional distribution of the $i$th coordinate given that the first $i-1$ coordinates are $x_{<i}=(x_1,...,x_{i-1})$.

We define POVM-based universal quantum inductive inference as follows:
\begin{definition}[POVM-Based Universal Quantum Inductive Inference]
    Let $n,s,t,\epsilon^{-1},\delta^{-1}\in\N$ be parameters.
    Let $\cL$ be a quantum algorithm that behaves as follows:
    \begin{enumerate}
        \item Take the following as input.
        \begin{itemize}
            \item Parameters $\mathsf{param}\coloneq (1^n,1^s,1^t,1^{\epsilon^{-1}},1^{\delta^{-1}})$.
            \item A collection $h_{<i}\coloneq (M_1,x_1,...,M_{i-1},x_{i-1})$, where each $M_j$ is (a description of) a POVM acting on $\regR_j$ and each $x_j$ is a classical string.
        \item A description of a target POVM $M_i$ acting on $\regR_i$.
    \end{itemize}
    \item Output a classical string $x_i$.
\end{enumerate}
Let $r:\N^3\to\N$ be a function.
We say that $\cL$ solves POVM-based universal quantum inductive inference with round complexity $r$ if, for all $m\ge r(s,\epsilon^{-1},\delta^{-1})$, all sequences of POVMs $M=(M_1,...,M_m)$, and all $mn$-qubit states $\rho_{\regR_1...\regR_m}\in\cC_{n,m}[s,t]$, the following holds:
\begin{align}
    \Pr_{i,x_{<i}} \left[ \SD\left( \cL(\mathsf{param},h_{<i},M_i), \cD^{M,\rho}_i(\cdot|x_{<i}) \right)\le\epsilon \right] \ge 1-\delta,
\end{align}
where the probability is taken over $i\gets[m]$ and $x_{<i}=(x_1,...,x_{i-1})\gets\cD^{M,\rho}_{<i}$.
\end{definition}
\begin{remark}
Our definition includes a natural quantum extension of classical
inductive inference.
In classical inductive inference, a classical algorithm generates
a sequence of strings $(x_1,\ldots,x_m)$, and the learner uses
$x_{<i}$ to predict the conditional distribution of $x_i$.
The most direct quantum extension lets a quantum algorithm generate
the same kind of classical sequence.
This task is included in our definition: the output strings can be
stored in computational-basis registers, and each $M_i$ can measure
$\regR_i$ in that basis.
For a general quantum source, the distribution of observed strings
also depends on the measurements performed on its registers.
We use the POVM-based formulation to capture these prediction tasks
for any specified sequence of POVMs, including measurements of
sources whose registers are entangled.
\end{remark}

Note that the POVM-based definition is a special case of \cref{def:uqii}.
Indeed, if we choose instruments $(\cI_1,...,\cI_m)$ such that $\cI_i=\{\cE^{(i)}_{x}\}_x$, where $\cE^{(i)}_x(\rho)=\Tr[E^{(i)}_x\rho]|0^n\rangle\langle0^n|$ for some POVM $\{E^{(i)}_x\}_x$, \cref{def:uqii} reduces to the POVM-based definition.
Therefore, we obtain the following information-theoretic upper bound on the round complexity of POVM-based universal quantum inductive inference as a corollary of \cref{thm:info_alg}.
\label{sec:info_povm}
\begin{corollary}\label{thm:povm-information-theoretic}
    There exists an algorithm that solves POVM-based universal quantum inductive inference with round complexity $O(s\epsilon^{-2}\delta^{-1})$.
\end{corollary}

\subsection{Average-Case Computational Hardness}
\label{sec:comp_povm}
Next, we introduce the notion of average-case POVM-based universal quantum inductive inference.
\begin{definition}[Average-Case POVM-based Universal Quantum Inductive Inference]
\label{def:avg_povm_uqii}
Let $\secp\in\N$ be a parameter and let $r:\N^4\to\N$ be a polynomial.
We say that a QPT algorithm $\cL$ solves average-case POVM-based universal quantum inductive inference with round complexity $r$ if, for all polynomials $m,n,s,t,t',p,q:\N\to\N$, all $m(\secp)\ge r(\secp,s(\secp),p(\secp),q(\secp))$, all $t'(\secp)$-time algorithms $\cG$, the following holds:
\begin{itemize}
    \item $\cG$ takes $1^\secp$ as input and outputs a Turing machine $T$ such that $|T|\le s(\secp)$ and $T^{t(\secp)}$ generates a circuit description of $\rho_{\regR_1...\regR_{m(\secp)}}^T$ and a sequence $M=(M_1,...,M_{m(\secp)})$ of POVMs, where each $M_i$ acts on $\regR_i$ and is described as a quantum circuit implementing it.
    \item For all sufficiently large $\secp$,
    \begin{align}
        \Pr_{i,T,M,x_{<i}} \left[ \SD \left(\cL(\param_\secp,h_{<i},M_i), \cD^{M,\rho^T}_i(\cdot|x_{<i}) \right) \le \frac{1}{p(\secp)} \right] \ge 1-\frac{1}{q(\secp)},
    \end{align}
    where $\param_\secp\coloneqq (1^\secp,1^{n(\secp)},1^{s(\secp)},1^{t(\secp)},1^{t'(\secp)},1^{p(\secp)},1^{q(\secp)})$ and $h_{<i}\coloneqq (M_1,x_1,...,M_{i-1},x_{i-1})$.
    The probability is taken over $i\gets[m(\secp)]$, $(T,M)\gets\cG(1^\secp)$, and $x_{<i}\gets \cD^{M,\rho^T}_{<i}$.
\end{itemize}
\end{definition}

We show that the average-case hardness of POVM-based universal quantum inductive inference is equivalent to the existence of infinitely-often OWPuzzs.
\begin{theorem}
\label{thm:owpuzz_povmhard}
    Infinitely-often OWPuzzs exist if and only if there exists no QPT algorithm that solves average-case POVM-based universal quantum inductive inference with round complexity $O((s(\secp)+g(\secp))p(\secp)^2q(\secp))$ for any superconstant function $g(\secp)=\omega(1)$.
\end{theorem}
\begin{proof}[Proof of \cref{thm:owpuzz_povmhard}]
The proof of the ``only if'' direction is essentially the same as that of \cref{thm:owpuzz_to_hard}.
Thus, it suffices to show the ``if'' direction.
We prove it by contraposition. 
Suppose that infinitely-often OWPuzzs do not exist. 
By the equivalence between infinitely-often OWPuzzs and infinitely-often distributional OWPuzzs (\cref{lem:distowpuzz_owpuzz}), infinitely-often distributional OWPuzzs do not exist.
Then, for any QPT algorithm $\Samp(1^N)\to(X,Y)$ that outputs classical strings $(X,Y)$, and any constant $c>0$, there exists a QPT algorithm $\mathsf{Inv}$ such that, for all sufficiently large $N$,
\begin{align}
    \SD\left( \{(X,Y)\}, \{(\mathsf{Inv}(1^N,Y),Y)\} \right) \le \frac{1}{N^c},
\end{align}
where $(X,Y)\gets\Samp(1^N)$.
Consider the following QPT algorithm $\cQ(1^N)$.
\begin{itemize}
    \item Take $1^N$ as input and let $b_{N}\in\bit^*$ be the binary representation of it.
    If $|b_{N}|\neq 7k+1$ for some $k\in\N$ or the first bit of $b_N$ is 0, then $\cQ$ outputs $\bot$.
    Otherwise, let $b_N=1\|b_1\|...\|b_{7}$, where $b_i\in\bit^k$.
    Let $\secp,n(\secp),s(\secp),t(\secp),t'(\secp),p(\secp),q(\secp)\in\N$ be natural numbers whose binary representations are $b_1,...,b_7$, respectively.
    If the input $N$ can not be parsed to natural numbers as above, $\cQ$ outputs $\bot$.
    \item Let $\Gamma\gets\bit^{N}$ and parse $\Gamma=1^{\ell}\|0\|c\|d$ for some $c\in\bit^\ell$ and $d\in\bit^{N-2\ell-1}$.
    Interpret $c$ as the description of a Turing machine $\cG_\Gamma$.
    If $\Gamma$ can not be parsed as above, $\cQ$ outputs $\bot$.
    \item Run $\cG_\Gamma(1^\secp)$ for $t'(\secp)$ steps to obtain a Turing machine $T$ and a sequence $(M'_1,...,M'_m)$ of POVMs.
    If $|T|>s(\secp)$ or $\cG_\Gamma(1^\secp)$ does not output valid descriptions, then $\cQ$ outputs $\bot$.
    Moreover, define $M=(M_1,...,M_{t'(\secp)})$ by
    \begin{align}
        M_{j} \coloneqq 
        \begin{cases}
            M'_j & \text{ if $j\le m$} \\
            I & \text{if $m <j\le t'(\secp)$}.
        \end{cases}
    \end{align}
    \item Run $T^{t(\secp)}$ to obtain a quantum circuit and run it to prepare $\rho^T_{\regR_1...\regR_{t'(\secp)}}$.
    \item Apply $M_1,...,M_{t'(\secp)}$ to registers $(\regR_1,...,\regR_{t'(\secp)})$, obtaining $(x_1,...,x_{t'(\secp)})$.
    \item Output $(M_1,x_1,...,M_{t'(\secp)},x_{t'(\secp)})$.
\end{itemize}
Moreover, let $\cQ_i(\cdot|M_1,x_1...,M_{i})$ be the conditional distribution of $x_i$ conditioned on $(M_1,x_1,...,M_i)$.
Let $\Samp$ be a QPT algorithm that on input $1^N$, samples $(M_1,x_1,...,M_{t'(\secp)},x_{t'(\secp)})\gets\cQ(1^N)$ and $i\gets[t'(\secp)]$, and outputs $X\coloneqq x_i$ and $Y\coloneqq (h_{<i},M_i)$, where $h_{<i}=(M_1,x_1,...,M_{i-1},x_{i-1})$.
Then, by the assumption that infinitely-often distributional OWPuzzs do not exist, there exists a QPT algorithm $\mathsf{Inv}$ such that, for all sufficiently large $N$,
\begin{align}
    \SD\left( \{(X,Y)\}_{(X,Y)\gets\Samp(1^N)}, \{(\mathsf{Inv}(1^N,Y),Y)\}_{(X,Y)\gets\Samp(1^N)} \right)
    \le \frac{1}{N^4}.
\end{align}
Now fix an arbitrary $t'(\secp)$-time algorithm $\cG$ and let $c_{\cG}$ denote the binary classical description of $\cG$.
Let $\widehat{\cG}\coloneqq 1^{|c_\cG|}\|0\|c_{\cG}$.
By the construction of $\Samp$, for all sufficiently large $N$, the probability that the randomly selected program $\Gamma$ in $\cQ$ begins with $\widehat{\cG}$ is $2^{-|\widehat{\cG}|}$.
Conditioned on this event, the distribution of $(T,M_1,x_1,...,M_{t'(\secp)},x_{t'(\secp)})$ corresponds to the sampling $(T,M_1,...,M_m)\gets\cG(1^\lambda)$, and $(x_1,...,x_{t'(\secp)})\gets \cD^{M,\rho^T}_{<t'(\secp)}$.
Therefore, for all sufficiently large $N$,
\begin{align}
    \frac{1}{N^4} 
    &\ge
    \SD\left( \{(X,Y)\}_{(X,Y)\gets\Samp(1^N)}, \{(\mathsf{Inv}(1^N,Y),Y)\}_{(X,Y)\gets\Samp(1^N)} \right) \\
    &\ge \frac{1}{2^{|\widehat{\cG}|}} \mathbb{E}_{i,T,M,x_{<i}} \left[ \SD\left( \cQ_i(\cdot|h_{<i},M_i), \mathsf{Inv}(1^N,h_{<i},M_i) \right) \right],
\end{align}
where in the last inequality, the expectation is taken over $i\gets[t'(\secp)]$, $(T,M'_1,...,M'_m)\gets\cG(1^\secp)$, and $x_{<i}\gets\cD^{M,\rho^T}_{<i}$.
When we sample $i\gets[m]$, 
\begin{align}
    &\mathbb {E}_{\substack{i\gets[m]\\T,M,x_{<i}}} \left[ \SD\left( \mathsf{Inv}(1^N,(h_{<i},M_i)), \cQ_i(\cdot|h_{<i},M_i) \right) \right] \\
    &= \frac{t'(\secp)}{t'(\secp)} \frac{1}{m} \sum_{i\in[m]} \mathbb {E}_{T,M,x_{<i}} \left[ \SD\left( \mathsf{Inv}(1^N,(h_{<i},M_i)), \cQ_i(\cdot|h_{<i},M_i) \right) \right]  \\
    &\le \frac{t'(\secp)}{m} \frac{1}{t'(\secp)} \sum_{i\in[t'(\secp)]} \mathbb {E}_{T,M,x_{<i}} \left[ \SD\left( \mathsf{Inv}(1^N,(h_{<i},M_i)), \cQ_i(\cdot|h_{<i},M_i) \right) \right]  \\
    &= \frac{t'(\secp)}{m} \mathbb {E}_{\substack{i\gets[t'(\secp)]\\T,M,x_{<i}}} \left[ \SD\left( \mathsf{Inv}(1^N,(h_{<i},M_i)), \cQ_i(\cdot|h_{<i},M_i) \right) \right] \\
    &\le \frac{2^{|\widehat{\cG}|}t'(\secp)}{N^4m}
\end{align}
holds for all sufficiently large $N$, where $(T,M'_1,...,M'_m)\gets \cG(1^\lambda)$, and $x_{<i}\gets \cD^{M,\rho^T}_{<i}$. 
By Markov's inequality,
\begin{align}
    \Pr_{i,T,M,x_{<i}}
    \left[ \SD\left( \mathsf{Inv}(1^N,(h_{<i},M_i)), \cQ_i(\cdot|h_{<i},M_i) \right) > \frac{1}{p(\secp)} \right] 
    \le \frac{2^{|\widehat{\cG}|}p(\secp)t'(\secp)}{N^4m}
\end{align}
holds for all sufficiently large $N$.
Because $N\ge \lambda+t'(\secp)+p(\secp)+q(\secp)$,
\begin{align}
    \frac{2^{|\widehat{\cG}|}p(\secp)t'(\secp)q(\secp)}{N^4} 
    \le \frac{2^{|\widehat{\cG}|}p(\secp)t'(\secp)q(\secp)}{(\secp+t'(\secp)+p(\secp)+q(\secp))^4} 
    \le 1
\end{align}
holds for all sufficiently large $\secp$.
Consequently, for all sufficiently large $\secp$,
\begin{align}
    \Pr_{i,T,M,x_{<i}}
    \left[ \SD\left( \mathsf{Inv}(1^N,(h_{<i},M_i)), \cQ_i(\cdot|h_{<i},M_i) \right) 
    > \frac{1}{p(\secp)} \right]
    \le \frac{1}{mq(\secp)}
    \le \frac{1}{q(\secp)}.
    \label{eq:universal_inv}
\end{align}
Later, we will show that 
\begin{align}
    \sum_{i\in[m]} \mathbb{E}_{i,T,M,x_{<i}} \left[ \mathrm{KL}(\cD^{M,\rho^T}_{i}(\cdot|x_{<i})\| \cQ_i(\cdot|h_{<i},M_i)) \right] \le s(\secp)+1+|\widehat{\cG}|.
    \label{eq:KL}
\end{align}
holds for all sufficiently large $\secp$.
We construct a QPT learning algorithm $\cL$ as follows:
\begin{enumerate}
    \item Take parameters $\mathsf{param}_\secp=(1^\secp,1^{n(\secp)},1^{s(\secp)},1^{t(\secp)},1^{t'(\secp)},1^{p(\secp)},1^{q(\secp)})$ and classical strings $h_{<i},M_i$ as input.
    \item Let $b'_1,...,b'_7$ be binary representations of $\secp,n(\secp),s(\secp),t(\secp),t'(\secp),2p(\secp),2q(\secp)$, where $k=\max\{|b'_i|\}$.
    Let $N\in\N$ be the natural number whose binary representation is $1\|b_1\|...\|b_7$, where $b_i=0^{k-|b'_i|}\|b'_i$.
    \item Run $\mathsf{Inv}(1^N,(h_{<i},M_i))$.
\end{enumerate}
Then, 
\begin{align}
    &\Pr_{i,T,M,x_{<i}}
    \left[ \SD\left( \cL(\param_\secp,h_{<i},M_i), \cD^{M,\rho^T}_i(\cdot|x_{<i}) \right) 
    > \frac{1}{p(\secp)} \right] \\
    &\le \Pr_{i,T,M,x_{<i}}\left[ \SD\left( \mathsf{Inv}(1^N,(h_{<i},M_i)), \cQ_i(\cdot|h_{<i},M_i) \right) 
    > \frac{1}{2p(\secp)} \right] \\
    &\qquad + \Pr_{i,T,M,x_{<i}} \left[ \SD\left( \cD^{M,\rho^T}_i(\cdot|x_{<i}), \cQ_i(\cdot|h_{<i},M_i) \right) > \frac{1}{2p(\secp)} \right] \\
    &\le \frac{1}{2q(\secp)}+ \frac{4p(\secp)^2 (s(\secp)+1+|\widehat{\cG}|)}{m}.
\end{align}
The last inequality follows from \cref{eq:universal_inv,eq:KL}.
Requiring that this probability is upper bounded by $1/q(\secp)$ implies $\frac{4p(\secp)^2 (s(\secp)+1+|\widehat{\cG}|)}{m}\le \frac{1}{2q(\secp)}$.
Therefore, we obtain 
\begin{align}
    m \ge 8 (s(\secp)+1+|\widehat{\cG}|)p(\secp)^2 q(\secp),
\end{align}
and the round complexity $O((s(\secp)+g(\secp))p(\secp)^2q(\secp))$ for any superconstant function $g(\secp)=\omega(1)$ is sufficient.
In the remaining part, we prove \cref{eq:KL}.
Note that for any $(M_1,x_1,...,M_m,x_m)$, 
\begin{align}
    &\Pr[(M_1,x_1,...,M_m,x_m)\gets\cQ(1^N)] \\
    &\ge 2^{-|\widehat{\cG}|} \Pr[(T,M_1,...,M_m)\gets\cG(1^\secp)\land(x_1,...,x_m)\gets\cD^{M,\rho^T}_{\le m}] \\
    &\ge 2^{-|\widehat{\cG}|} \Pr[T\gets\cG(1^\secp)] \Pr[(M_1,...,M_m)\gets\cG(1^\secp)\land(x_1,...,x_m)\gets\cD^{M,\rho^T}_{\le m}|T]. 
\end{align}
By the chain rule for KL divergence, we have
\begin{align}
    &\sum_{i\in[m]} \mathbb{E}_{T,M,x_{<i}} \left[ \mathrm{KL}\left(\cD^{M,\rho^T}_i(\cdot|x_{<i}) \| \cQ_i(\cdot|h_{<i},M_i) \right) \right] \\
    &\le \mathbb{E}_T \Bigl[ \sum_{M_1,x_1,...,M_m,x_m}\Pr[(M_1,...M_m)\gets\cG(1^\secp)\land (x_1,...,x_m)\gets\cD^{M,\rho^T}_{\le m}|T] \\
    & \qquad \times \log\frac{\Pr[(M_1,...M_m)\gets\cG(1^\secp)\land (x_1,...,x_m)\gets\cD^{M,\rho^T}_{\le m}|T]}{\Pr[(M_1,x_1,...,M_m,x_m)\gets\cQ(1^N)]} \Bigr] \\
    &\le \mathbb{E}_T (|\widehat{\cG}|-\log\Pr[T\gets \cG(1^\secp)]) \\
    &\le s(\secp)+1+|\widehat{\cG}|.
\end{align}
In the last inequality, we used the fact that $|T|\le s(\secp)$ and its entropy $\mathbb{E}_T[-\log\Pr[T\gets\cG(1^\secp)]]$ is at most $s(\secp)+1$.
Thus, we obtain \cref{eq:KL} and complete the proof.
\end{proof}

\section{Non-i.i.d. Tomography}
\label{sec:tomography}

In the previous sections, we studied quantum inductive inference and its POVM-based variant.
We now consider state tomography in the non-i.i.d.\ setting.
The task differs from universal quantum inductive inference in \cref{sec:uqii} in the following two respects:
\begin{itemize}
    \item The learner outputs a classical description of the conditional state of the future register $\regR_i$.
    \item The learner chooses the measurements to perform on the preceding registers $(\regR_1,...,\regR_{i-1})$.
\end{itemize}
The target is the marginal state on $\regR_i$ conditioned on the learner's measurement record.
We first recall the source class and then define the learning task.
Let $n,m,s,t\in\N$ be parameters, and for each $i\in[m]$, let $\regR_i$ be an
$n$-qubit register.
We define the set
\begin{align}
    \cC_{n,m}[s,t]
    :=\left\{
        \rho_{\regR_1...\regR_m}:
        \begin{gathered}
            \text{There exists a classical Turing machine $T$ with $|T|\le s$} \\
            \text{ such that $T^t$ outputs a circuit preparing $\rho_{\regR_1...\regR_m}$}
        \end{gathered}
    \right\}
\end{align}
of $mn$-qubit states whose preparation circuit descriptions are generated
by Turing machines described by at most $s$ bits within $t$ steps.

\begin{definition}[Non-i.i.d.\ State Tomography]
Let $r:\mathbb N^4\to\mathbb N$ be a function.
Let $\mathcal L$ be a quantum algorithm that behaves as follows:
\begin{enumerate}
    \item Take the following as input:
    \begin{itemize}
        \item Parameters
        $\mathsf{param}:=(1^s,1^t,1^n,
        1^{\epsilon^{-1}},1^{\delta^{-1}})$.
        \item A single copy of an $n(i-1)$-qubit state on $(\regR_1,...,\regR_{i-1})$.
    \end{itemize}
    The round $i$ is determined by the number of input registers.

    \item Perform a measurement procedure on $(\regR_1,...,\regR_{i-1})$ and record its
    outcome $x$.
    For fixed $\mathsf{param}$ and $i$, let
    $M:=\{E_x\}_{x\in\mathcal X}$ be the POVM induced by the entire
    procedure, including any adaptive measurement choices.
    The outcome $x$ records all measurement outcomes and internal
    random choices, including those used to determine the output
    state description.

    \item Output $h:=(M,x)$ and a classical description of an
    $n$-qubit state $\sigma$.
    Here $M$ is represented by a classical program implementing the
    measurement procedure, and the description of $\sigma$ is
    determined by $\mathsf{param}$, $i$, and $x$.
\end{enumerate}

The target conditional state of $\regR_i$ is
\begin{align}
    \rho_{\regR_i\mid h}
    :=\frac{
        \operatorname{Tr}_{\regR_1...\regR_{i-1}}\left[
            (E_x\otimes I_{\regR_i})\rho_{\regR_1...\regR_i}
        \right]
    }{\operatorname{Tr}(E_x\rho_{\regR_1...\regR_{i-1}})}.
\end{align}

We say that $\mathcal L$ solves non-i.i.d.\ state tomography with
round complexity $r$ if, for all
$n,s,t,\epsilon^{-1},\delta^{-1}\in\mathbb N$,
all $m\ge r(n,s,\epsilon^{-1},\delta^{-1})$, and all
$\rho_{\regR_1... \regR_m}\in\mathcal C_{n,m}[s,t]$,
\begin{align}
    \Pr\left[
        \TD(\sigma,\rho_{\regR_i\mid h})\le\epsilon
        \;:\;
        \begin{array}{l}
            i\leftarrow[m],\\
            (h,\sigma)\leftarrow\mathcal L(\mathsf{param},\rho_{\regR_1...\regR_{i-1}})
        \end{array}
    \right]
    \ge 1-\delta.
\end{align}
\end{definition}

We now give an upper bound on the round complexity.

\begin{theorem}
\label{thm:state_tomography_improve}
There exists an algorithm that solves non-i.i.d.\ state tomography
with round complexity
\begin{align}
    r(n,s,\epsilon^{-1},\delta^{-1})
    =O\left(\frac{sn^2}{\epsilon^4\delta^2}\right).
\end{align}
\end{theorem}

\begin{proof}[Proof of \cref{thm:state_tomography_improve}]
We first consider the setting in which there is a prior distribution over Turing machines that generate the circuit descriptions of quantum states and construct a tomography algorithm for each prior distribution with a bounded error guarantee.
We then use multiplicative weight updates to combine these distribution-dependent algorithm into a single universal learner with an error bound for any source Turing machine.

Let $\cT$ be the set of Turing machines described by at most $s$ bits.
For each $T\in\cT$, let $\sigma^T_{\regR_1...\regR_m}$ be the state whose preparation circuit is generated by $T^t$.
Let $p$ be a distribution over $\cT$ and define
\begin{align}
    \Omega_{\regA\regR_1...\regR_i} \coloneqq \sum_{T\in\cT} p(T) |T\rangle\langle T|_{\regA} \otimes \sigma^T_{\regR_1...\regR_i}.
\end{align}
For an $n$-qubit state $\omega$, let $\cD(\omega)\coloneqq\omega/2+I/2^{n+1}$ be the depolarizing channel and let $S(\omega)\coloneqq -\Tr(\omega\log\omega)$ be the base-two von Neumann entropy.

For any $\xi>0$, let $M=\{E_x\}_x$ be a POVM acting on $(\regR_1,...,\regR_{i-1})$ that satisfies
\begin{align}
    \sum_x \Tr[E_x\Omega_{\regR_1...\regR_{i-1}}] S\left( \cD(\Omega_{\regR_i|M,x}) \right)
        \le \inf_{M'=\{E'_x\}_x} \sum_x \Tr[E'_x\Omega_{\regR_1...\regR_{i-1}}] S\left( \cD(\Omega_{\regR_i|M',x}) \right) + \xi,
    \label{eq:measurement_choice}
\end{align}
where
\begin{align}
    \Omega_{\regR_i|M,x} \coloneqq \frac{\Tr_{\regR_1...\regR_{i-1}}[(E_x\otimes I_{\regR_i})\Omega_{\regR_1...\regR_i}]} {\Tr(E_x\Omega_{\regR_1...\regR_{i-1}})}.
\end{align}
Let $\cA$ be a quantum algorithm that takes $n(i-1)$-qubit state over registers $(\regR_1,...,\regR_{i-1})$, applies $M$ to obtain a result $x$, and outputs a classical description of $\Omega_{\regR_i|M,x}$.

We now bound its average error.
Let $\cR_{\Omega,\cN}:\regR_1...\regR_{i-1}\to\regA\regR_1...\regR_{i-1}$ be the averaged rotated Petz recovery map of \cref{lem:Petz_recover} with respect to a reference state $\Omega_{\regA\regR_1...\regR_{i-1}}$ and a channel $\cN\coloneqq \Tr_{\regA}$ tracing out $\regA$.
Then, $\cR_{\Omega,\cN}(\Omega_{\regR_1...\regR_{i-1}})=\Omega_{\regA\regR_1...\regR_{i-1}}$ because the Petz map exactly recovers the target state when it is given as the reference state, and by \cref{lem:Petz_recover},
\begin{align}
    &D(\Omega_{\regA\regR_1...\regR_i}\|\Omega_{\regA\regR_1...\regR_{i-1}}\otimes \Omega_{\regR_i}) - D(\Omega_{\regR_1...\regR_i}\|\Omega_{\regR_1...\regR_{i-1}}\otimes\Omega_{\regR_i}) \\
    &\ge -2\log F(\Omega_{\regA\regR_1...\regR_i}, (\cR_{\Omega,\cN}\otimes \mathrm{id}_{\regR_i})\Omega_{\regR_1...\regR_i}).
    \label{eq:tomo_1}
\end{align}
Moreover, by the definition of quantum mutual information, we have
\begin{align}
    &D(\Omega_{\regA\regR_1...\regR_i}\|\Omega_{\regA\regR_1...\regR_{i-1}}\otimes \Omega_{\regR_i}) - D(\Omega_{\regR_1...\regR_i}\|\Omega_{\regR_1...\regR_{i-1}}\otimes\Omega_{\regR_i}) \\
    &= I(\regA\regR_1...\regR_{i-1}:\regR_i)_\Omega - I(\regR_1...\regR_{i-1}:\regR_i)_\Omega \\
    &= I(\regA:\regR_i|\regR_1...\regR_{i-1})_\Omega.
    \label{eq:tomo_2}
\end{align}
Combining \cref{eq:tomo_1,eq:tomo_2} implies
\begin{align}
    \TD(\Omega_{\regA\regR_1...\regR_i},(\cR_{\Omega,\cN}\otimes \mathrm{id}_{\regR_i})\Omega_{\regR_1...\regR_i} )
    &\le \sqrt{1- F(\Omega_{\regA\regR_1...\regR_i},(\cR_{\Omega,\cN}\otimes \mathrm{id}_{\regR_i})\Omega_{\regR_1...\regR_i} )^2} \\
    &\le \sqrt{1-2^{-I(\regA:\regR_i|\regR_1...\regR_{i-1})_\Omega}} \\
    &\le \sqrt{I(\regA:\regR_i|\regR_1...\regR_{i-1})_\Omega}.
    \label{eq:recovery_error}
\end{align}
Define $\widetilde\Omega_{\regA\regR_1...\regR_i}\coloneqq (\cR_{\Omega,\cN}\otimes\mathrm{id}_{\regR_i})(\Omega_{\regR_1...\regR_i})$.
Let $\regX$ denote the classical register recording the outcome of $M$
and $\cM:\regA\regR_1...\regR_i\to\regA\regX\regR_i$ denote the measurement channel corresponding to $M$ defined by
\begin{align}
    \cM(\Omega_{\regA\regR_1...\regR_i}) = \sum_{T,x} p(T)\Tr[E_x\sigma^T_{\regR_1...\regR_{i-1}}] |T,x\rangle\langle T,x|_{\regA\regX} \otimes \sigma^T_{\regR_i|M,x}.
    \label{eq:measurement_true}
\end{align}
Moreover, since $\widetilde{\Omega}_{\regA\regR_1...\regR_{i-1}}=\Omega_{\regA\regR_1...\regR_{i-1}}$ from \cref{lem:Petz_recover}, we have
\begin{align}
    \cM(\widetilde{\Omega}_{\regA\regR_1...\regR_i}) = \sum_{T,x} p(T)\Tr[E_x\sigma^T_{\regR_1...\regR_{i-1}}] |T,x\rangle\langle T,x|_{\regA\regX} \otimes \widetilde{\sigma}^T_{\regR_i|M,x}
    \label{eq:measurement_fake}
\end{align}
for some $n$-qubit state $\widetilde{\sigma}^T_{\regR_i|M,x}$.
Consider the POVM measurement $M'$ on $(\regR_1,...,\regR_{i-1})$ that 
behaves as follows:
\begin{enumerate}
    \item Apply $\cR_{\Omega,\cN}$ to obtain a state on registers $(\regA,\regR_1,...,\regR_{i-1})$.
    \item Measure $\regA$ in the computational basis and obtain $T$.
    \item Apply $M$ to $(\regR_1,...,\regR_{i-1})$ and obtain $x$.
    \item Output $T$ and $x$.
\end{enumerate}
When we apply $M'$ to $\Omega_{\regR_1...\regR_{i}}$ and obtain the measurement result $(T,x)$, the conditioned state on $\regR_i$ is $\widetilde{\sigma}^T_{\regR_i|M,x}$.
Because $M$ satisfies \cref{eq:measurement_choice}, we have
\begin{align}
    \sum_{x} \Tr[E_x\Omega_{\regR_1...\regR_{i-1}}] S\left( \cD(\Omega_{\regR_i|M,x}) \right) 
    \le \sum_{T,x} p(T)\Tr[E_x\sigma^T_{\regR_1...\regR_{i-1}}] S\left( \cD(\widetilde{\sigma}^T_{\regR_i|M,x}) \right) + \xi
\end{align}
By \cref{eq:measurement_true,eq:measurement_fake},
\begin{align}
    \sum_{T,x}p(T)\Tr[E_x\sigma^T_{\regR_1...\regR_{i-1}}]\TD(\sigma^T_{\regR_i|M,x},\widetilde{\sigma}^T_{\regR_i|M,x})
    &=\TD(\cM(\Omega_{\regA\regR_1...\regR_i}),\cM(\widetilde{\Omega}_{\regA\regR_1...\regR_i}))\\
    &\le\TD(\Omega_{\regA\regR_1...\regR_i},\widetilde\Omega_{\regA\regR_1...\regR_i})\\
    &\le\sqrt{I(\regA:\regR_i|\regR_1...\regR_{i-1})_\Omega}.
\end{align}
In the last inequality, we used \cref{eq:recovery_error}.
Let $\cD(\regR_i)$ denote the resulting register after applying $\cD$ to $\regR_i$.
Then, by the definition of conditional quantum mutual information,
\begin{align}
    &I(\regA:\cD(\regR_i)|\regX)_{\cM(\Omega)} \\
    &= \sum_{x} \Tr[E_x\Omega_{\regR_1...\regR_{i-1}}] S\left( \cD(\Omega_{\regR_i|M,x}) \right) - \sum_{T,x}p(T)\Tr[E_x\sigma^T_{\regR_1...\regR_{i-1}}] S\left( \cD(\sigma^T_{\regR_i|M,x}) \right) \\
    &\le \sum_{T,x} p(T)\Tr[E_x\sigma^T_{\regR_1...\regR_{i-1}}] S\left( \cD(\widetilde{\sigma}^T_{\regR_i|M,x}) \right) - \sum_{T,x}p(T)\Tr[E_x\sigma^T_{\regR_1...\regR_{i-1}}] S\left( \cD(\sigma^T_{\regR_i|M,x}) \right) + \xi \\
    &\le \frac{n+1}{2} \sum_{T,x}p(T)\Tr[E_x\sigma^T_{\regR_1...\regR_{i-1}}] \TD\left( \sigma^T_{\regR_i|M,x}, \widetilde{\sigma}^T_{\regR_i|M,x} \right) + \xi \quad (\text{By \cref{lem:depolarize}}.)\\
    &\le \frac{n+1}{2} \sqrt{I(\regA:\regR_i|\regR_1...\regR_{i-1})_\Omega} + \xi.
\end{align}
Thus,
\begin{align}
    \mathbb E_{T\gets p}\mathbb E_{h\gets\cA_p(\sigma^T_{\regR_1...\regR_{i-1}})}
    D\left(\cD(\sigma^T_{\regR_i|h})\middle\|\cD(\Omega_{\regR_i|h})\right)
    &=I(\regA:\cD(\regR_i)|\regX)_{\cM(\Omega)} \\
    &\le\frac{n+1}{2}\sqrt{I(\regA:\regR_i|\regR_1...\regR_{i-1})_\Omega}+\xi,
\end{align}
where $h=(M,x)$.
By Pinsker's inequality and $\TD(\cD(\omega),\cD(\nu))=\TD(\omega,\nu)/2$, we obtain
\begin{align}
    \mathbb{E}_{T\gets p} \mathbb{E}_{h\gets\cA_p(\sigma^T_{\regR_1...\regR_{i-1}})} \TD(\sigma^T_{\regR_i|h},\Omega_{\regR_i|h})^2 
    &\le 2\ln2 \mathbb{E}_{T\gets p} \mathbb{E}_{h\gets\cA_p(\sigma^T_{\regR_1...\regR_{i-1}})} D\left(\cD(\sigma^T_{\regR_i|h})\middle\|\cD(\Omega_{\regR_i|h})\right)\\
    &\le(n+1)\sqrt{I(\regA:\regR_i|\regR_1...\regR_{i-1})_\Omega}+2\xi.
    \label{eq:improve_bayes_round}
\end{align}
The last inequality uses $\ln2<1$.

For each $k\in\N$, we consider the set $J_k\coloneqq\{2^k,2^k+1,...,2^{k+1}-1\}$.
By the chain rule of mutual information,
\begin{align}
    \sum_{i\in J_k} I(\regA:\regR_i|\regR_1...\regR_{i-1})_\Omega
    &\le\sum_{i\in[2^{k+1}-1]} I(\regA:\regR_i|\regR_1...\regR_{i-1})_\Omega\\
    &=I(\regA:\regR_1...\regR_{2^{k+1}-1})_\Omega \\
    &\le H(p)
    \le\log |\cT|
    \le s+1
\end{align}
holds for all $k$.
Summing \cref{eq:improve_bayes_round} over $i\in J_k$ and applying Cauchy--Schwarz inequality gives
\begin{align}
    \sum_{i\in J_k} \mathbb{E}_{T\gets p} \mathbb{E}_{h\gets\cA_p(\sigma^T_{\regR_1...\regR_{i-1}})} \TD(\sigma^T_{\regR_i|h},\Omega_{\regR_i|h})^2 
    &\le (n+1) \sum_{i\in J_k} \sqrt{I(\regA:\regR_i|\regR_1...\regR_{i-1})_\Omega} + 2|J_k| \xi \\
    &\le (n+1) \sqrt{|J_k|\sum_{i\in J_k} I(\regA:\regR_i|\regR_1...\regR_{i-1})_\Omega} + 2|J_k|\xi \\
    &\le (n+1) \sqrt{|J_k|(s+1)} + 2|J_k|\xi.
    \label{eq:block_error}
\end{align}
This inequality bounds the average squared trace distance error of the learner $\cA_p$ that outputs a description of $\Omega_{\regR_i|h}$ for any prior distribution $p$.
We next turn the preceding average-case bound into a single learner
that works for every $T\in\mathcal T$.
Fix $k$ and $J_k$.
For an algorithm $\cA$ that takes $n(i-1)$-qubit state as input and outputs a pair of measurement and its result $h=(M,x)$ and a classical description of an $n$-qubit state, define
\begin{align}
  F_T(\cA):=
  \sum_{i\in J_k}
  \mathbb E_{(h,\sigma)\gets\cA(\sigma^T_{\regR_1...\regR_{i-1}})}
  \left[
    \TD \left(\sigma^T_{\regR_i\mid h}, \sigma\right)^2
  \right].
\end{align}
Since trace distance is at most $1$, $0\le F_T(\cA)\le |J_k|$.
\cref{eq:block_error} implies that, for every probability distribution $p$ on $\cT$ and every $\xi>0$, there exists an algorithm $\cA$ such that
\begin{equation}
  \sum_{T\in\cT}p(T) F_T(\cA) \le (n+1)\sqrt{|J_k|(s+1)} + 2|J_k|\xi.
\end{equation}

Choose $q\in\N$ satisfying
\begin{align}
  q \ge \frac{8\sqrt{|J_k|(s+1)}}{n+1}.
  \label{eq:b2-mw-iterations}
\end{align}
Let $w_{T,1}\coloneqq 1$ for all $T\in\cT$.
For $j\in[q]$, set
\begin{align}
    W_j \coloneqq \sum_{T\in\cT} w_{T,j},
    \text{ and }
    p_j(T) \coloneqq \frac{w_{T,j}}{W_j}.
\end{align}
For each $j\in[q]$, let $\cA_j$ be an algorithm that satisfies
\begin{equation}
  \sum_{T\in\cT} p_j(T) F_T(\cA_j) < \frac{5}{4}(n+1)\sqrt{|J_k|(s+1)}.
\end{equation}
Each $w_{T,j}$ is defined by the following recurrence relation:
\begin{align}
    w_{T,j+1}\coloneqq w_{T,j} \left (1+\frac{F_{T}(\cA_j)}{2|J_k|}\right).
\end{align}
Then
\begin{align}
    W_{j+1} 
    &= \sum_{T\in\cT} w_{T,j+1} \\
    &= \sum_{T\in\cT} w_{T,j} \left( 1+\frac{F_T(\cA_j)}{2|J_k|} \right) \\
    &= W_{j} + W_j \sum_{T\in\cT} p_j(T) \frac{F_{T}(\cA_j)}{2|J_k|} \\
    &\le W_j \left( 1+\frac{5(n+1)\sqrt{|J_k|(s+1)}}{8|J_k|} \right).
\end{align}

Using $W_1=|\cT|$ and $\ln(1+u)\le u$ for $u\ge 0$, we obtain
\begin{equation}
  \ln W_{q+1}
  \le \ln |\cT| + \frac{5 q (n+1)\sqrt{|J_k|(s+1)}}{8|J_k|}.
  \label{eq:b2-mw-total-weight}
\end{equation}
Here and below, $\ln$ denotes the natural logarithm.

On the other hand, for $0\le u\le 1/2$,
\begin{align}
  \ln(1+u)
  =\int_0^u\frac{dv}{1+v}
  \ge u-\frac{u^2}{2}
  \ge \frac34 u.
\end{align}
Since $0\le F_{T}(\cA_j) \le |J_k|$, this gives, for every $T$,
\begin{align}
  \ln w_{T,q+1}
  &= \sum_{j\in[q]} \ln\left(1+\frac{F_{T}(\cA_j)}{2|J_k|} \right) \\
  &\ge \frac{3}{8|J_k|} \sum_{j\in[q]} F_{T}(\cA_j).
  \label{eq:b2-mw-individual-weight}
\end{align}
Combining $w_{T,q+1}\le W_{q+1}$ with \cref{eq:b2-mw-total-weight,eq:b2-mw-individual-weight}, we obtain
\begin{equation}
  \frac{1}{q} \sum_{j\in[q]} F_T(\cA_j)
  \le \frac{5}{3} (n+1)\sqrt{|J_k|(s+1)} + \frac{8|J_k| \ln |\cT|}{3\cdot q}
  \le 2 (n+1)\sqrt{|J_k|(s+1)}
\end{equation}
for all $T\in\cT$.
The last inequality follows from $|\cT|<2^{s+1}$ and \cref{eq:b2-mw-iterations}.

Define $C_k$ by choosing $j\gets[q]$ and running $\cA_j$.
Let $M_i^{(j)}=\{E_{i,x}^{(j)}\}_x$ be the POVM of $\cA_j$ at round $i$.
The POVM of $C_k$ is
\begin{align}
  M_i^{\mathrm{mix}} \coloneqq \left\{\frac{1}{q} E_{i,x}^{(j)}\right\}_{j,x}.
\end{align}
If $C_k$ samples $j$ and $\cA_j$ produces outcome $x$ and a classical description of $\sigma$,
then $C_k$ outputs $h=(M_i^{\mathrm{mix}},(j,x))$ and $\sigma$.
For every outcome of positive probability,
\begin{align}
  \sigma^T_{\regR_i | M_i^{\mathrm{mix}},(j,x)}
  &= \frac{ \Tr_{\regR_1...\regR_{i-1}} \left[ ((E_{i,x}^{(j)}/q) \otimes I_{\regR_i}) \sigma^T_{\regR_1...\regR_{i}} \right] }{ \Tr \left[(E_{i,x}^{(j)}/q)\sigma^T_{\regR_1...\regR_{i-1}}\right] } \\
  &= \sigma^T_{\regR_i | M_i^{(j)},x}.
\end{align}
Thus
\begin{align}
  F_T(C_k)
  =
  \frac{1}{q} \sum_{j\in[q]} F_T(\cA_j)
  \le
  2(n+1)\sqrt{|J_k|(s+1)}
\end{align}
for all $T\in\cT$.

We finally describe the learner $\cL$ as follows:
\begin{enumerate}
    \item Take $\mathsf{param}$ and a state $\rho_{\regR_1...\regR_{i-1}}$ as input.
    \item Set $k:=\lfloor\log i\rfloor$ and construct $C_k$ as above. More precisely, sample $j\gets[q]$ and run $\cA_j$ on $\rho_{\regR_1...\regR_{i-1}}$.
    \item Let $M_i^{(j)}$, $x$, $\sigma$ be outcomes of $\cA_j$, where $M_i^{(j)}$ is the description of a POVM, $x$ is a classical string, and $\sigma$ is a classical description of an $n$-qubit state.
    \item Output $h=(M_i^{\mathrm{mix}},(j,x))$ and $\sigma$.
\end{enumerate}

Fix the target state $\rho_{\regR_1..\regR_m}\in\cC_{n,m}[s,t]$,
and let $T^*\in\mathcal T$ be the program generating the target state.
For all $i\le m$, $\sigma^{T^*}_{\regR_1...\regR_i}=\rho_{\regR_1...\regR_i}$.
We obtain
\begin{align}
  \sum_{i=1}^m \mathbb{E}_{(h,\sigma)\gets\cL(\mathsf{param},\rho_{\regR_1...\regR_{i-1}})} \left[ \TD(\rho_{\regR_i|h},\sigma)^2 \right]
  &= \sum_{k=0}^{\lfloor\log m\rfloor} \sum_{i\in J_k\cap[m]} \mathbb{E}_{(h,\sigma)\gets\cL(\mathsf{param},\rho_{\regR_1...\regR_{i-1}})} \left[ \TD(\rho_{\regR_i|h},\sigma)^2 \right] \\
  &\le \sum_{k=0}^{\lfloor\log m\rfloor} F_{T^*}(C_k) \\
  &\le 2(n+1)\sqrt{s+1}\sum_{k=0}^{\lfloor\log m\rfloor} 2^{k/2}\\
  &\le 2(2+\sqrt2)(n+1)\sqrt{(s+1)m}.
\end{align}
For the last inequality, we used
\begin{align}
  \sum_{k=0}^{\lfloor\log m\rfloor} 2^{k/2}
  = \frac{2^{(\lfloor\log m\rfloor+1)/2}-1}{\sqrt{2}-1}
  \le (2+\sqrt{2}) 2^{\lfloor\log m\rfloor/2}
  \le (2+\sqrt2)\sqrt{m}.
\end{align}
For a uniformly random $i\in[m]$, Markov's inequality gives
\begin{align}
  \Pr_{i,h,\sigma} \left[ \TD \left(\rho_{\regR_i|h},\sigma\right)>\epsilon \right]
  &\le \frac{1}{m\epsilon^2}\sum_{i=1}^m \mathbb{E}_{(h,\sigma)\gets\cL(\mathsf{param},\rho_{\regR_1...\regR_{i-1}})} \left[ \TD(\rho_{\regR_i|h},\sigma)^2 \right] \\
  &\le \frac{2(2+\sqrt{2})(n+1)}{\epsilon^2}\sqrt{\frac{s+1}{m}}.
\end{align}
Requiring that the right-hand side is at most $\delta$ gives
\begin{align}
  m\ge \frac{ 4(2+\sqrt{2})^2(n+1)^2(s+1) }{ \epsilon^4\delta^2 }
\end{align}
and therefore the round complexity
\begin{align}
  r(n,s,\epsilon^{-1},\delta^{-1}) = O\left(\frac{sn^2}{\epsilon^4\delta^2}\right)
\end{align}
is sufficient.
This completes the proof.
\end{proof}

\section{AI Disclosure}
We used generative AI tools, OpenAI GPT-5.6 Sol and GPT-6 Astra, to assist with improving the presentation of the manuscript and refining proof arguments in \cref{sec:uqii,sec:POVM,sec:tomography}. All proofs were independently verified by the authors, who take full responsibility for the content of the manuscript.

\ifnum\anonymous=1
\else
\section{Acknowledgements}
YS is supported by JSPS KAKENHI, Grant Number JP26KJ1421.
\fi

\appendix

\section{Lower bound on round complexity of Inductive Inference}\label{sec:lower_bound}

We show that the round complexity of inductive inference given in \cref{thm:info_alg} is optimal up to constant factors. This improves the lower bound of $\Omega(s/\epsilon^2+s/\delta)$ shown in~\cite[Proposition~B.1]{HN26}.

\begin{theorem}[Lower bound on round complexity]
    \label{thm:qii_lower_bound}
    Universal quantum inductive inference requires
    $\Omega(s/(\epsilon^2\delta))$ rounds.
    More specifically, there exists a constant $C>0$ such that
    this lower bound holds for all
    $\epsilon^{-1},\delta^{-1}\in\mathbb{N}$ satisfying
    $0<\epsilon\le 1/64$, $0<\delta\le 1/8$, and
    $s\ge C\log_2(2/(\epsilon\delta))$.
    It holds even for computationally unbounded learners,
    product states with $n=1$, and computational-basis instruments.
\end{theorem}

\begin{proof}[Proof of \cref{thm:qii_lower_bound}]
We prove a lower bound for classical inductive inference.
The learner is given a prefix $x_{<i}=(x_1,\ldots,x_{i-1})$
generated by an unknown randomized Turing machine and must
predict the conditional distribution of the next bit.

The idea is to hardwire $k$ bits into the Turing machine and
use each bit to determine the bias of a separate block of
outputs. The bias is small enough that observing a block
does not reliably reveal the corresponding hardwired bit.
However, accurately predicting the next output in that block
would reveal this bit.

Fix an arbitrary learner $\mathcal{L}$ and parameters
$s,\epsilon,\delta$ satisfying the conditions in the theorem.
Let $k=\lfloor as\rfloor$, where $a>0$ is a sufficiently small
constant, and set
\begin{align}
    \ell:=\left\lfloor\frac{1}{2048\epsilon^2}\right\rfloor,
    \qquad
    m:=\left\lfloor\frac{k\ell}{4\delta}\right\rfloor.
    \label{eq:qii_lb_parameters}
\end{align}
Thus, $\ell=\Theta(\epsilon^{-2})$ and $m\ge k\ell$.

Choose two rational numbers $p_0,p_1$ which satisfy
\begin{align}
    \frac12-3\epsilon\le p_0\le\frac12-2\epsilon,
    \qquad
    \frac12+2\epsilon\le p_1\le\frac12+3\epsilon.
    \label{eq:qii_lb_bias}
\end{align}

For each $\alpha=(\alpha_1,\ldots,\alpha_k)\in\{0,1\}^k$,
define a randomized Turing machine $T_\alpha$ with output
$(X_1,\ldots,X_m)$. Its first $k\ell$ output bits form $k$
consecutive blocks of length $\ell$. The first block consists
of $X_1,\ldots,X_\ell$, the second of
$X_{\ell+1},\ldots,X_{2\ell}$, and so on.
The $j$th block uses the hardwired bit $\alpha_j$:
each of its $\ell$ bits is $1$ with probability $p_{\alpha_j}$,
using fresh randomness for every bit. In other words,
\begin{align}
    \Pr_{(X_1,\ldots,X_m)\leftarrow T_\alpha}
    \left[X_{(j-1)\ell+u}=1\right]
    =p_{\alpha_j}
    \qquad (j\in[k],\ u\in[\ell]).
    \label{eq:qii_lb_process}
\end{align}
Here, $j$ is the block number, $u$ is the position within
that block, and $i=(j-1)\ell+u$ is the position in the entire
output. After these $k$ blocks, $T_\alpha$ appends
$m-k\ell$ zeros. For fixed $\alpha$, all output bits are
independent.
The string $\alpha$ is fixed in the description of $T_\alpha$;
it is not sampled during the execution of $T_\alpha$.

The description of $T_\alpha$ consists of this fixed procedure
and the values $\alpha,k,p_0,p_1,\epsilon^{-1},\delta^{-1}$.
The machine computes $\ell$ and $m$ from these values. Hence,
\begin{align}
    |T_\alpha|
    =O\!\left(k+\log k+\log\epsilon^{-1}
                    +\log\delta^{-1}\right)
    \le s,
\end{align}
where the last inequality holds by choosing $a$ sufficiently
small and $C$ sufficiently large. All these machines have a
common time bound $t$ determined by $k,\epsilon,\delta$.
In particular, the parameters given to the learner do not
depend on $\alpha$.

We now analyze the learner's prediction.
For a fixed prefix $x_{<i}$, the probability
$\Pr[1\leftarrow\mathcal{L}(\mathsf{param},x_{<i})]$
is taken over the learner's internal randomness.
At round $i=(j-1)\ell+u$ in the $j$th block, independence gives
\begin{align}
    \Pr_{(X_1,\ldots,X_m)\leftarrow T_\alpha}
    [X_i=1\mid X_{<i}=x_{<i}]
    =p_{\alpha_j}.
\end{align}
For distributions on one bit, statistical distance is the
absolute difference between their probabilities of $1$.
The prediction is therefore $\epsilon$-accurate exactly when
\begin{align}
    \left|
        \Pr[1\leftarrow\mathcal{L}(\mathsf{param},x_{<i})]
        -p_{\alpha_j}
    \right|\le\epsilon.
    \label{eq:qii_lb_success}
\end{align}

Our goal is to prove
\begin{align}
    \Pr_{\substack{
        \alpha\leftarrow\{0,1\}^k,\ i\leftarrow[m]\\
        (x_1,\ldots,x_m)\leftarrow T_\alpha
    }}\!\left[
        \mathcal{L}(\mathsf{param},x_{<i})
        \text{ is not $\epsilon$-accurate}
    \right]
    \ge\frac{3k\ell}{8m}
    \ge\frac32\delta>\delta.
    \label{eq:qii_lb_average_failure}
\end{align}
Here, $i$ is independent of $\alpha$ and the machine's
randomness, and accuracy is measured against the conditional
distribution under the fixed machine $T_\alpha$.
This inequality implies that there exists a fixed
$\alpha\in\{0,1\}^k$ for which the learner's failure probability
exceeds $\delta$. Since $|T_\alpha|\le s$ and all the machines
have the common time bound $t$, the required guarantee then
fails at length $m$. Consequently, any valid round complexity
must satisfy
\begin{align}
    r(s,\epsilon^{-1},\delta^{-1})
    >m=\Omega\!\left(\frac{s}{\epsilon^2\delta}\right).
\end{align}

To prove~\cref{eq:qii_lb_average_failure}, it suffices to show
that, for every $j\in[k]$ and $u\in[\ell]$, writing
$i=(j-1)\ell+u$, we have
\begin{align}
    \Pr_{\substack{
        \alpha\leftarrow\{0,1\}^k\\
        (x_1,\ldots,x_m)\leftarrow T_\alpha
    }}\!\left[
        \left|
            \Pr[1\leftarrow\mathcal{L}(\mathsf{param},x_{<i})]
            -p_{\alpha_j}
        \right|>\epsilon
    \right]\ge\frac38.
    \label{eq:qii_lb_each_round}
\end{align}
Indeed, a uniformly chosen round lies in the first $k\ell$
rounds with probability $k\ell/m$. Applying this estimate
to each of those rounds gives
\cref{eq:qii_lb_average_failure}.

We now prove~\cref{eq:qii_lb_each_round}.
Fix $j\in[k]$, $u\in[\ell]$, and $i=(j-1)\ell+u$.
In the following comparison, probabilities are taken in
the experiment
\begin{align*}
    \alpha\leftarrow\{0,1\}^k,
    \qquad
    (X_1,\ldots,X_m)\leftarrow T_\alpha.
\end{align*}
We compare the distribution of the observed prefix
conditioned on $\alpha_j=0$ with the one conditioned on
$\alpha_j=1$. We will show that
\begin{align}
    &\frac12\sum_{x_{<i}\in\{0,1\}^{i-1}}
        \left|
            \Pr[X_{<i}=x_{<i}\mid\alpha_j=0]
            -\Pr[X_{<i}=x_{<i}\mid\alpha_j=1]
        \right|
    \le\frac14.
    \label{eq:qii_lb_history_distance}
\end{align}

First, we explain why~\cref{eq:qii_lb_history_distance}
implies~\cref{eq:qii_lb_each_round}.
Consider the rule that outputs $1$ if
$\Pr[1\leftarrow\mathcal{L}(\mathsf{param},x_{<i})]>1/2$
and $0$ otherwise.
By~\cref{eq:qii_lb_bias}, this rule correctly determines
$\alpha_j$ whenever the prediction is accurate.
Its error probability satisfies
\begin{align}
    &\frac12
    \Pr_{\substack{
        \alpha\leftarrow\{0,1\}^k\\
        (x_1,\ldots,x_m)\leftarrow T_\alpha
    }}\!\left[
        \Pr[1\leftarrow\mathcal{L}(\mathsf{param},x_{<i})]
        >\frac12
        \,\middle|\,\alpha_j=0
    \right]
        \nonumber\\
    &\quad+\frac12
    \Pr_{\substack{
        \alpha\leftarrow\{0,1\}^k\\
        (x_1,\ldots,x_m)\leftarrow T_\alpha
    }}\!\left[
        \Pr[1\leftarrow\mathcal{L}(\mathsf{param},x_{<i})]
        \le\frac12
        \,\middle|\,\alpha_j=1
    \right]
        \nonumber\\
    &\quad=\frac12-\frac12\Biggl(
        \Pr_{\substack{
            \alpha\leftarrow\{0,1\}^k\\
            (x_1,\ldots,x_m)\leftarrow T_\alpha
        }}\!\left[
            \Pr[1\leftarrow\mathcal{L}(\mathsf{param},x_{<i})]
            >\frac12
            \,\middle|\,\alpha_j=1
        \right]
        \nonumber\\
    &\hspace{95pt}
        -\Pr_{\substack{
            \alpha\leftarrow\{0,1\}^k\\
            (x_1,\ldots,x_m)\leftarrow T_\alpha
        }}\!\left[
            \Pr[1\leftarrow\mathcal{L}(\mathsf{param},x_{<i})]
            >\frac12
            \,\middle|\,\alpha_j=0
        \right]
    \Biggr)
        \nonumber\\
    &\quad\ge\frac12-\frac12\cdot\frac14
        =\frac38.
    \label{eq:qii_lb_decision_error}
\end{align}
The inequality follows from~\cref{eq:qii_lb_history_distance}.
An incorrect decision implies failure of~\cref{eq:qii_lb_success},
which proves~\cref{eq:qii_lb_each_round}.
This argument does not require computing the learner's
output probability.

To prove~\cref{eq:qii_lb_history_distance}, we bound the
relative entropy between the distributions of $X_{<i}$
conditioned on $\alpha_j=0$ and $\alpha_j=1$:
\begin{align}
    &\sum_{x_{<i}\in\{0,1\}^{i-1}}
        \Pr[X_{<i}=x_{<i}\mid\alpha_j=0]
        \ln\frac{
            \Pr[X_{<i}=x_{<i}\mid\alpha_j=0]
        }{
            \Pr[X_{<i}=x_{<i}\mid\alpha_j=1]
        }
        \nonumber\\
    &\quad=(u-1)\left(
        p_0\ln\frac{p_0}{p_1}
        +(1-p_0)\ln\frac{1-p_0}{1-p_1}
    \right)
        \nonumber\\
    &\quad\le (u-1)\frac{(p_0-p_1)^2}{p_1(1-p_1)}
        \nonumber\\
    &\quad\le256(u-1)\epsilon^2
        \le\frac18.
    \label{eq:qii_lb_relative_entropy}
\end{align}
For the equality, the earlier blocks have the same distribution
under both conditions and are independent of the current block.
Thus, only the $u-1$ observed bits of the current block contribute.
The first inequality follows from $\ln z\le z-1$.
For the second, we used $p_0,p_1\in[1/4,3/4]$ and
$|p_0-p_1|\le6\epsilon$.
The final inequality follows from $u-1<\ell$ and the choice
of $\ell$ in~\cref{eq:qii_lb_parameters}.
Pinsker's inequality gives~\cref{eq:qii_lb_history_distance}.
This proves~\cref{eq:qii_lb_each_round} and hence
\cref{eq:qii_lb_average_failure}.

Finally, the same lower bound applies to universal quantum
inductive inference. A classical Turing machine can output
a circuit generating these independent bits on one-qubit
registers. With computational-basis instruments, the
post-measurement prefix is completely determined by
$x_{<i}$ and provides no additional information.
An $\epsilon$-accurate joint prediction must also predict
the next outcome with error at most $\epsilon$, by
monotonicity of trace distance under discarding registers.
For the fixed finite gate set, it suffices to prepare the
registers independently so that the two outcome probabilities
lie in the intervals in~\cref{eq:qii_lb_bias}.
The classical argument above uses only these intervals and
independence, so it applies without requiring exact
implementation of the chosen rational probabilities.
A fixed compilation procedure preserves the description-length
estimate, and we can again choose a common time bound.
This completes the proof.
\end{proof}

\ifnum\submission=0
\bibliographystyle{alpha} 
\else
\bibliographystyle{splncs04}
\fi
\bibliography{abbrev0,crypto,reference}

@string{ieee =                  {IEEE}}

@string{springer =              "Springer"}

@string{dagstuhl =              "Schloss Dagstuhl - Leibniz-Zentrum fuer Informatik"}

@string{acm =                   "Association for Computing Machinery"}

@inproceedings{CL21,
  author    = {Kai-Min Chung and Han-Hsuan Lin},
  title     = {Sample Efficient Algorithms for Learning Quantum Channels
               in {PAC} Model and the Approximate State Discrimination Problem},
  booktitle = {16th Conference on the Theory of Quantum Computation,
               Communication and Cryptography ({TQC} 2021)},
  editor    = {Min-Hsiu Hsieh},
  series    = {Leibniz International Proceedings in Informatics (LIPIcs)},
  volume    = {197},
  pages     = {3:1--3:22},
  publisher = {Schloss Dagstuhl -- Leibniz-Zentrum f{\"u}r Informatik},
  year      = {2021},
  doi       = {10.4230/LIPIcs.TQC.2021.3}
}

@article{BO24,
  author  = {Costin B{\u{a}}descu and Ryan O'Donnell},
  title   = {Improved Quantum Data Analysis},
  journal = {TheoretiCS},
  volume  = {3},
  pages   = {1--34},
  note    = {Article 7},
  year    = {2024},
  doi     = {10.46298/theoretics.24.7}
}

@misc{QZ26,
      title={Impersonating Quantum Secrets over Classical Channels}, 
      author={Luowen Qian and Mark Zhandry},
      year={2026},
      eprint={2601.01058},
      archivePrefix={arXiv},
      primaryClass={quant-ph},
      url={https://arxiv.org/abs/2601.01058}, 
}

@inproceedings{CCCGHJL26,
  author    = {Bruno Cavalar and Boyang Chen and Andrea Coladangelo
               and Matthew Gray and Zihan Hu and Zhengfeng Ji
               and Xingjian Li},
  title     = {A Meta-complexity Characterization of Minimal Quantum Cryptography},
  booktitle = {Proceedings of the 58th Annual ACM Symposium on Theory of Computing (STOC)},
  pages     = {675--686},
  publisher = {ACM},
  year      = {2026},
  doi       = {10.1145/3798129.3800783}
}

@inproceedings{Perrier26,
  author    = {Elija Perrier},
  title     = {Quantum {AIXI}: Universal Intelligence via Quantum Information},
  booktitle = {Artificial General Intelligence},
  series    = {Lecture Notes in Computer Science},
  volume    = {16058},
  pages     = {58--70},
  publisher = {Springer},
  year      = {2026},
  doi       = {10.1007/978-3-032-00800-8_6}
}

@techreport{Sol60,
  author      = {Ray J. Solomonoff},
  title       = {A Preliminary Report on a General Theory of Inductive Inference},
  institution = {Zator Company},
  number      = {ZTB-138},
  address     = {Cambridge, MA},
  year        = {1960},
  month       = nov,
  url         = {https://www.raysolomonoff.com/publications/z138.pdf}
}

@article{Sol64II,
  author  = {Ray J. Solomonoff},
  title   = {A Formal Theory of Inductive Inference. {Part II}},
  journal = {Information and Control},
  volume  = {7},
  number  = {2},
  pages   = {224--254},
  year    = {1964},
  doi     = {10.1016/S0019-9958(64)90131-7}
}

@inproceedings{IL90,
  author    = {Russell Impagliazzo and Leonid A. Levin},
  title     = {No Better Ways to Generate Hard {NP} Instances
               than Picking Uniformly at Random},
  booktitle = {Proceedings of the 31st Annual Symposium on
               Foundations of Computer Science ({FOCS})},
  pages     = {812--821},
  publisher = {IEEE Computer Society},
  year      = {1990},
  doi       = {10.1109/FSCS.1990.89604},
  url       = {https://doi.org/10.1109/FSCS.1990.89604}
}

@inproceedings{BFKL93,
  author    = {Avrim Blum and Merrick L. Furst and
               Michael J. Kearns and Richard J. Lipton},
  title     = {Cryptographic Primitives Based on Hard Learning Problems},
  editor    = {Douglas R. Stinson},
  booktitle = {Advances in Cryptology---{CRYPTO} '93},
  series    = {Lecture Notes in Computer Science},
  volume    = {773},
  pages     = {278--291},
  publisher = {Springer},
  year      = {1994},
  doi       = {10.1007/3-540-48329-2_24},
  url       = {https://doi.org/10.1007/3-540-48329-2_24}
}

@article{Wil15,
    author = {Wilde, Mark M.},
    title = {Recoverability in quantum information theory},
    journal = {Proceedings of the Royal Society A: Mathematical, Physical and Engineering Sciences},
    volume = {471},
    number = {2182},
    pages = {20150338},
    year = {2015},
    month = {10},
    issn = {1364-5021},
    doi = {10.1098/rspa.2015.0338},
    url = {https://doi.org/10.1098/rspa.2015.0338},
    eprint = {https://royalsocietypublishing.org/rspa/article-pdf/doi/10.1098/rspa.2015.0338/367681/rspa.2015.0338.pdf},
}

@article{VNHUS11,
  author  = {Joel Veness and Kee Siong Ng and Marcus Hutter and William Uther and David Silver},
  title   = {A {Monte-Carlo} {AIXI} Approximation},
  journal = {Journal of Artificial Intelligence Research},
  year    = {2011},
  doi     = {10.1613/jair.3125}
}

@article{Kol65,
  author  = {Kolmogorov, Andrei N.},
  title   = {Three Approaches to the Quantitative Definition of Information},
  journal = {Problems of Information Transmission},
  volume  = {1},
  number  = {1},
  pages   = {1--7},
  year    = {1965}
}

@article{Cha66,
  author  = {Chaitin, Gregory J.},
  title   = {On the Length of Programs for Computing Finite Binary Sequences},
  journal = {Journal of the ACM},
  volume  = {13},
  number  = {4},
  pages   = {547--569},
  year    = {1966},
  doi     = {10.1145/321356.321363}
}

@article{Svo96,
  author  = {Svozil, Karl},
  title   = {Quantum Algorithmic Information Theory},
  journal = {Journal of Universal Computer Science},
  volume  = {2},
  number  = {5},
  pages   = {311--346},
  year    = {1996},
  doi     = {10.3217/jucs-002-05-0311}
}

@article{Vit01,
  author  = {Vit{\'a}nyi, Paul M. B.},
  title   = {Quantum {Kolmogorov} Complexity Based on Classical Descriptions},
  journal = {IEEE Transactions on Information Theory},
  volume  = {47},
  number  = {6},
  pages   = {2464--2479},
  year    = {2001},
  doi     = {10.1109/18.945258}
}

@article{MB05,
  author  = {Mora, Caterina E. and Briegel, Hans J.},
  title   = {Algorithmic Complexity and Entanglement of Quantum States},
  journal = {Physical Review Letters},
  volume  = {95},
  number  = {20},
  pages   = {200503},
  year    = {2005},
  doi     = {10.1103/PhysRevLett.95.200503}
}

@article{BDL01,
  author  = {Berthiaume, Andr{\'e} and van Dam, Wim and Laplante, Sophie},
  title   = {Quantum {Kolmogorov} Complexity},
  journal = {Journal of Computer and System Sciences},
  volume  = {63},
  number  = {2},
  pages   = {201--221},
  year    = {2001},
  doi     = {10.1006/jcss.2001.1765}
}

@article{Gac01,
  author  = {G{\'a}cs, Peter},
  title   = {Quantum Algorithmic Entropy},
  journal = {Journal of Physics A: Mathematical and General},
  volume  = {34},
  number  = {35},
  pages   = {6859--6880},
  year    = {2001},
  doi     = {10.1088/0305-4470/34/35/312}
}

@article{Val84,
  author  = {Leslie G. Valiant},
  title   = {A Theory of the Learnable},
  journal = {Communications of the ACM},
  volume  = {27},
  number  = {11},
  pages   = {1134--1142},
  year    = {1984},
  doi     = {10.1145/1968.1972}
}

@inproceedings{NR06,
  author    = {Moni Naor and Guy N. Rothblum},
  title     = {Learning to Impersonate},
  booktitle = {Proceedings of the 23rd International Conference
               on Machine Learning ({ICML} 2006)},
  pages     = {649--656},
  publisher = {ACM},
  year      = {2006},
  doi       = {10.1145/1143844.1143926}
}

@inproceedings{HN23,
  author    = {Shuichi Hirahara and Mikito Nanashima},
  title     = {Learning in {Pessiland} via Inductive Inference},
  booktitle = {2023 IEEE 64th Annual Symposium on Foundations
               of Computer Science ({FOCS})},
  pages     = {447--457},
  publisher = {IEEE},
  year      = {2023},
  doi       = {10.1109/FOCS57990.2023.00033}
}

@inproceedings{HN26,
  author    = {Shuichi Hirahara and Mikito Nanashima},
  title     = {Complexity-Theoretic Universal Inductive Inference},
  booktitle = {Proceedings of the 58th Annual ACM Symposium
               on Theory of Computing ({STOC} 2026)},
  pages     = {377--385},
  publisher = {ACM},
  year      = {2026},
  doi       = {10.1145/3798129.3800757}
}

@inproceedings{HHM25,
  author    = {Taiga Hiroka and Min-Hsiu Hsieh and Tomoyuki Morimae},
  title     = {Hardness of Quantum Distribution Learning
               and Quantum Cryptography},
  booktitle = {Advances in Cryptology -- {ASIACRYPT} 2026},
  year      = {2026},
  note      = {To appear},
  eprint    = {2507.01292},
  archivePrefix = {arXiv},
  primaryClass  = {quant-ph},
  url       = {https://arxiv.org/abs/2507.01292}
}

@misc{HH24,
  author    = {Taiga Hiroka and Min-Hsiu Hsieh},
  title     = {Computational Complexity of Learning
               Efficiently Generatable Pure States},
  year      = {2024},
  eprint    = {2410.04373},
  archivePrefix = {arXiv},
  primaryClass  = {quant-ph},
  url       = {https://arxiv.org/abs/2410.04373}
}

@inproceedings{FGSY25,
  author    = {Bill Fefferman and Soumik Ghosh and Makrand Sinha
               and Henry Yuen},
  title     = {The Hardness of Learning Quantum Circuits
               and Its Cryptographic Applications},
  booktitle = {17th Innovations in Theoretical Computer Science
               Conference ({ITCS} 2026)},
  series    = {Leibniz International Proceedings in Informatics (LIPIcs)},
  volume    = {362},
  pages     = {56:1--56:21},
  publisher = {Schloss Dagstuhl -- Leibniz-Zentrum f{\"u}r Informatik},
  year      = {2026},
  doi       = {10.4230/LIPIcs.ITCS.2026.56}
}

@misc{CL26,
  author    = {Alexandru Cojocaru and Laura Lewis},
  title     = {Equivalence Between Average-Case Hardness of Learning
               and Cryptography for Mixed Quantum States},
  year      = {2026},
  eprint    = {2608.14331},
  archivePrefix = {arXiv},
  primaryClass  = {quant-ph},
  url       = {https://arxiv.org/abs/2608.14331}
}

@inproceedings{BHHP24,
  author    = {John Bostanci and Jonas Haferkamp and Dominik Hangleiter
               and Alexander Poremba},
  title     = {Efficient Quantum Pseudorandomness
               from {Hamiltonian} Phase States},
  booktitle = {20th Conference on the Theory of Quantum Computation,
               Communication and Cryptography ({TQC} 2025)},
  series    = {Leibniz International Proceedings in Informatics (LIPIcs)},
  volume    = {350},
  pages     = {9:1--9:18},
  publisher = {Schloss Dagstuhl -- Leibniz-Zentrum f{\"u}r Informatik},
  year      = {2025},
  doi       = {10.4230/LIPIcs.TQC.2025.9}
}

@inproceedings{PQS25,
  author    = {Alexander Poremba and Yihui Quek and Peter Shor},
  title     = {The Learning Stabilizers with Noise Problem},
  booktitle = {17th Innovations in Theoretical Computer Science
               Conference ({ITCS} 2026)},
  series    = {Leibniz International Proceedings in Informatics (LIPIcs)},
  volume    = {362},
  pages     = {108:1--108:19},
  publisher = {Schloss Dagstuhl -- Leibniz-Zentrum f{\"u}r Informatik},
  year      = {2026},
  doi       = {10.4230/LIPIcs.ITCS.2026.108}
}

@inproceedings{CGGH25,
  author    = {Bruno P. Cavalar and Eli Goldin and Matthew Gray
               and Peter Hall},
  title     = {A Meta-complexity Characterization of Quantum Cryptography},
  booktitle = {Advances in Cryptology -- {EUROCRYPT} 2025, Part VII},
  series    = {Lecture Notes in Computer Science},
  volume    = {15607},
  pages     = {82--107},
  publisher = {Springer},
  year      = {2025},
  doi       = {10.1007/978-3-031-91098-2_4}
}

@inproceedings{HM25,
  author    = {Taiga Hiroka and Tomoyuki Morimae},
  title     = {Quantum Cryptography and Meta-Complexity},
  booktitle = {Advances in Cryptology -- {CRYPTO} 2025},
  series    = {Lecture Notes in Computer Science},
  volume    = {16001},
  pages     = {545--574},
  publisher = {Springer},
  year      = {2025},
  doi       = {10.1007/978-3-032-01878-6_18}
}

@inproceedings{QRZ25,
  author    = {Luowen Qian and Justin Raizes and Mark Zhandry},
  title     = {Hard Quantum Extrapolations in Quantum Cryptography},
  booktitle = {Advances in Cryptology -- {EUROCRYPT} 2025, Part VII},
  series    = {Lecture Notes in Computer Science},
  volume    = {15607},
  pages     = {53--81},
  publisher = {Springer},
  year      = {2025},
  doi       = {10.1007/978-3-031-91098-2_3}
}

@article{HM15,
  author        = {Hayashi, Masahito and Morimae, Tomoyuki},
  title         = {Verifiable Measurement-Only Blind Quantum Computing
                   with Stabilizer Testing},
  journal       = {Physical Review Letters},
  volume        = {115},
  pages         = {220502},
  year          = {2015},
  doi           = {10.1103/PhysRevLett.115.220502},
  eprint        = {1505.07535},
  archivePrefix = {arXiv},
  primaryClass  = {quant-ph},
  url           = {https://arxiv.org/abs/1505.07535}
}

@article{MTH17,
  author        = {Morimae, Tomoyuki and Takeuchi, Yuki
                   and Hayashi, Masahito},
  title         = {Verification of Hypergraph States},
  journal       = {Physical Review A},
  volume        = {96},
  pages         = {062321},
  year          = {2017},
  doi           = {10.1103/PhysRevA.96.062321},
  eprint        = {1701.05688},
  archivePrefix = {arXiv},
  primaryClass  = {quant-ph},
  url           = {https://arxiv.org/abs/1701.05688}
}

@article{TM18,
  author        = {Takeuchi, Yuki and Morimae, Tomoyuki},
  title         = {Verification of Many-Qubit States},
  journal       = {Physical Review X},
  volume        = {8},
  pages         = {021060},
  year          = {2018},
  doi           = {10.1103/PhysRevX.8.021060},
  eprint        = {1709.07575},
  archivePrefix = {arXiv},
  primaryClass  = {quant-ph},
  url           = {https://arxiv.org/abs/1709.07575}
}

@article{TMMMF19,
  author        = {Takeuchi, Yuki and Mantri, Atul
                   and Morimae, Tomoyuki and Mizutani, Akihiro
                   and Fitzsimons, Joseph F.},
  title         = {Resource-Efficient Verification of Quantum Computing
                   Using {Serfling}'s Bound},
  journal       = {npj Quantum Information},
  volume        = {5},
  pages         = {27},
  year          = {2019},
  doi           = {10.1038/s41534-019-0142-2},
  eprint        = {1806.09138},
  archivePrefix = {arXiv},
  primaryClass  = {quant-ph},
  url           = {https://arxiv.org/abs/1806.09138}
}

@article{ZH19a,
  author        = {Zhu, Huangjun and Hayashi, Masahito},
  title         = {Efficient Verification of Hypergraph States},
  journal       = {Physical Review Applied},
  volume        = {12},
  pages         = {054047},
  year          = {2019},
  doi           = {10.1103/PhysRevApplied.12.054047},
  eprint        = {1806.05565},
  archivePrefix = {arXiv},
  primaryClass  = {quant-ph},
  url           = {https://arxiv.org/abs/1806.05565}
}

@article{ZH19b,
  author        = {Zhu, Huangjun and Hayashi, Masahito},
  title         = {Efficient Verification of Pure Quantum States
                   in the Adversarial Scenario},
  journal       = {Physical Review Letters},
  volume        = {123},
  pages         = {260504},
  year          = {2019},
  doi           = {10.1103/PhysRevLett.123.260504},
  eprint        = {1909.01900},
  archivePrefix = {arXiv},
  primaryClass  = {quant-ph},
  url           = {https://arxiv.org/abs/1909.01900}
}

@article{ZH19c,
  author        = {Zhu, Huangjun and Hayashi, Masahito},
  title         = {General Framework for Verifying Pure Quantum States
                   in the Adversarial Scenario},
  journal       = {Physical Review A},
  volume        = {100},
  pages         = {062335},
  year          = {2019},
  doi           = {10.1103/PhysRevA.100.062335},
  eprint        = {1909.01943},
  archivePrefix = {arXiv},
  primaryClass  = {quant-ph},
  url           = {https://arxiv.org/abs/1909.01943}
}

@article{MK20,
  author        = {Markham, Damian and Krause, Alexandra},
  title         = {A Simple Protocol for Certifying Graph States
                   and Applications in Quantum Networks},
  journal       = {Cryptography},
  volume        = {4},
  number        = {1},
  pages         = {3},
  year          = {2020},
  doi           = {10.3390/cryptography4010003},
  eprint        = {1801.05057},
  archivePrefix = {arXiv},
  primaryClass  = {quant-ph},
  url           = {https://arxiv.org/abs/1801.05057}
}

@article{LZH23,
  author        = {Li, Zihao and Zhu, Huangjun and Hayashi, Masahito},
  title         = {Robust and Efficient Verification of Graph States
                   in Blind Measurement-Based Quantum Computation},
  journal       = {npj Quantum Information},
  volume        = {9},
  pages         = {115},
  year          = {2023},
  doi           = {10.1038/s41534-023-00783-9},
  eprint        = {2305.10742},
  archivePrefix = {arXiv},
  primaryClass  = {quant-ph},
  url           = {https://arxiv.org/abs/2305.10742}
}

@misc{AGMOC25,
  author        = {Abdul Sater, Sami and Garnier, Maxime
                   and Martinez, Thierry and Ollivier, Harold
                   and Chabaud, Ulysse},
  title         = {Efficient Certification of Intractable Quantum States
                   with Few {Pauli} Measurements},
  year          = {2025},
  eprint        = {2511.07300},
  archivePrefix = {arXiv},
  primaryClass  = {quant-ph},
  note          = {arXiv:2511.07300},
  url           = {https://arxiv.org/abs/2511.07300}
}

@misc{PFMO25,
  author        = {De Palma, Giacomo and Fanizza, Marco
                   and Mowry, Connor and O'Donnell, Ryan},
  title         = {Non-iid Hypothesis Testing: From Classical to Quantum},
  year          = {2025},
  eprint        = {2510.06147},
  archivePrefix = {arXiv},
  primaryClass  = {quant-ph},
  note          = {arXiv:2510.06147},
  url           = {https://arxiv.org/abs/2510.06147}
}

@misc{CGGHHM25,
  author        = {Cavalar, Bruno and Goldin, Eli and Gray, Matthew
                   and Hiroka, Taiga and Hsieh, Min-Hsiu
                   and Morimae, Tomoyuki},
  title         = {Cryptographic Conditions for Efficient Testing
                   of Distributions and Quantum States},
  year          = {2025},
  eprint        = {2510.05028},
  archivePrefix = {arXiv},
  primaryClass  = {quant-ph},
  note          = {arXiv:2510.05028},
  url           = {https://arxiv.org/abs/2510.05028}
}

@article{BH17,
  author  = {Brand{\~a}o, Fernando G. S. L. and Harrow, Aram W.},
  title   = {Quantum {de Finetti} Theorems Under Local Measurements
             with Applications},
  journal = {Communications in Mathematical Physics},
  volume  = {353},
  number  = {2},
  pages   = {469--506},
  year    = {2017},
  doi     = {10.1007/s00220-017-2880-3},
  url     = {https://arxiv.org/abs/1210.6367}
}

@article{FQR24,
  author  = {Fanizza, Marco and Quek, Yihui and Rosati, Matteo},
  title   = {Learning Quantum Processes Without Input Control},
  journal = {PRX Quantum},
  volume  = {5},
  pages   = {020367},
  year    = {2024},
  doi     = {10.1103/PRXQuantum.5.020367},
  url     = {https://arxiv.org/abs/2211.05005}
}

@misc{Zam26a,
  author        = {Zambrano, Leonardo},
  title         = {Quantum tomography for non-iid sources},
  year          = {2026},
  eprint        = {2602.22057},
  archivePrefix = {arXiv},
  primaryClass  = {quant-ph},
  note          = {arXiv:2602.22057},
  url           = {https://arxiv.org/abs/2602.22057}
}

@misc{Zam26b,
  author        = {Zambrano, Leonardo},
  title         = {Classical shadows for non-iid quantum sources},
  year          = {2026},
  eprint        = {2603.05137},
  archivePrefix = {arXiv},
  primaryClass  = {quant-ph},
  note          = {arXiv:2603.05137},
  url           = {https://arxiv.org/abs/2603.05137}
}

@misc{ZSK26,
  author        = {Yanbao Zhang and Akshay Seshadri and Emanuel Knill},
  title         = {An efficient method for spot-checking quantum
                   properties with sequential trials},
  year          = {2026},
  eprint        = {2602.08114},
  archivePrefix = {arXiv},
  primaryClass  = {quant-ph},
  url           = {https://arxiv.org/abs/2602.08114}
}

@misc{NZ26,
  author        = {Mariana Navarro and Leonardo Zambrano},
  title         = {Certifying quantum states without independence
                   assumptions},
  year          = {2026},
  eprint        = {2606.31913},
  archivePrefix = {arXiv},
  primaryClass  = {quant-ph},
  url           = {https://arxiv.org/abs/2606.31913}
}

@article{Sol78,
  author  = {Ray J. Solomonoff},
  title   = {Complexity-Based Induction Systems:
             Comparisons and Convergence Theorems},
  journal = {IEEE Transactions on Information Theory},
  volume  = {24},
  number  = {4},
  pages   = {422--432},
  year    = {1978},
  doi     = {10.1109/TIT.1978.1055913},
  url     = {https://raysolomonoff.com/publications/solo1.pdf}
}

@article{FKMO24,
  author        = {Fawzi, Omar and Kueng, Richard and Markham, Damian and Oufkir, Aadil},
  title         = {Learning Properties of Quantum States without the {IID} Assumption},
  journal       = {Nature Communications},
  volume        = {15},
  number        = {1},
  pages         = {9677},
  year          = {2024},
  doi           = {10.1038/s41467-024-53765-6},
  eprint        = {2401.16922},
  archivePrefix = {arXiv},
  primaryClass  = {quant-ph},
  url           = {https://arxiv.org/abs/2401.16922}
}

@article{CR12,
  author        = {Christandl, Matthias and Renner, Renato},
  title         = {Reliable Quantum State Tomography},
  journal       = {Physical Review Letters},
  volume        = {109},
  number        = {12},
  pages         = {120403},
  year          = {2012},
  doi           = {10.1103/PhysRevLett.109.120403},
  eprint        = {1108.5329},
  archivePrefix = {arXiv},
  primaryClass  = {quant-ph},
  url           = {https://arxiv.org/abs/1108.5329}
}

@article{Sol64,
  author        = {Solomonoff, Ray J.},
  title         = {A Formal Theory of Inductive Inference. {Part I}},
  journal       = {Information and Control},
  volume        = {7},
  number        = {1},
  pages         = {1--22},
  year          = {1964},
  doi           = {10.1016/S0019-9958(64)90223-2},
  url           = {https://doi.org/10.1016/S0019-9958(64)90223-2}
}

@ARTICLE{JRSWW18,
  title     = "Universal recovery maps and approximate sufficiency of quantum
               relative entropy",
  author    = "Junge, Marius and Renner, Renato and Sutter, David and Wilde,
               Mark M and Winter, Andreas",
  journal   = "Ann. Henri Poincaré",
  publisher = "Springer Science and Business Media LLC",
  volume    =  19,
  number    =  10,
  pages     = "2955--2978",
  month     =  oct,
  year      =  2018,
  language  = "en"
}

@Article{Petz86,
author={Petz, D{\'e}nes},
title={Sufficient subalgebras and the relative entropy of states of a von {N}eumann algebra},
journal={Communications in Mathematical Physics},
year={1986},
month={Mar},
day={01},
volume={105},
number={1},
pages={123-131},
issn={1432-0916},
doi={10.1007/BF01212345},
url={https://doi.org/10.1007/BF01212345}
}

@article{Petz88,
    author = {Petz, D{\'e}nes},
    title = {SUFFICIENCY OF CHANNELS OVER VON {N}EUMANN ALGEBRAS},
    journal = {The Quarterly Journal of Mathematics},
    volume = {39},
    number = {1},
    pages = {97-108},
    year = {1988},
    month = {03},
    issn = {0033-5606},
    doi = {10.1093/qmath/39.1.97},
    url = {https://doi.org/10.1093/qmath/39.1.97},
    eprint = {https://academic.oup.com/qjmath/article-pdf/39/1/97/4559225/39-1-97.pdf},
}

\end{document}